\documentclass[a4paper,11pt,reqno]{article}
\usepackage[margin=25mm]{geometry}
\usepackage{amsmath, amssymb, amsthm}
\usepackage{mathtools}
\usepackage{dsfont} 
\usepackage{enumitem}

\numberwithin{equation}{section} 

\newcommand{\R}{\mathbb{R}}

\newcommand{\bN}{\mathbb{N}}

\newcommand{\cB}{\mathcal{B}}

\newcommand{\cF}{\mathcal{F}}
\newcommand{\cH}{\mathcal{H}}

\newcommand{\cN}{\mathcal{N}}

\newcommand{\dd}{\mathrm{d}}

\newcommand{\fa}{\mathfrak{a}}
\newcommand{\hc}{\mathrm{h.c.}}
\newcommand{\vphi}{\varphi}
\newcommand{\vep}{\varepsilon}

\newcommand{\1}{\mathds 1}

\newcommand{\weyl}{W_{\hspace{-.06cm} N_0}}
\newcommand{\weyleta}{W_{\hspace{-.06cm} N_0(1+\eta) \,}}

\newcommand{\supp}{{\rm supp}}
\renewcommand{\geq}{\geqslant}
\renewcommand{\leq}{\leqslant}
\renewcommand{\hat}{\widehat}
\renewcommand{\tilde}{\widetilde}
\newcommand{\Tr}{\operatorname{Tr}}
\newcommand{\tr}{\operatorname{Tr}}

\newcommand{\nn}{\nonumber}

\usepackage{braket}

\usepackage{xcolor}

\usepackage{scalerel,stackengine} 
\stackMath
\newcommand\reallywidehat[1]{%
\savestack{\tmpbox}{\stretchto{%
  \scaleto{%
    \scalerel*[\widthof{\ensuremath{#1}}]{\kern-.6pt\bigwedge\kern-.6pt}%
    {\rule[-\textheight/2]{1ex}{\textheight}}
  }{\textheight}%
}{0.5ex}}%
\stackon[1pt]{#1}{\tmpbox}%
}

\usepackage{xstring} 

\usepackage[
  backend=biber,
  url=false,         
  doi=false,         
  isbn=false,        
  eprint=false,       
  maxnames=7,        
  giveninits=true,  
]{biblatex}

\usepackage[hidelinks]{hyperref} 

\DeclareFieldFormat{journaltitle}{%
  \ifboolexpr{ test {\iffieldundef{shortjournal}} }
    {\mkbibemph{#1}}
    {\mkbibemph{\thefield{shortjournal}}}%
}

\DeclareFieldFormat*{volume}{\mkbibbold{#1}}   
\DeclareFieldFormat*{number}{\mkbibbold{#1}}   

 \renewbibmacro{in:}{\addcomma\space}
 
\DeclareNameAlias{family-given}{author}  

\AtEveryBibitem{%
  \clearfield{month}%
    \clearfield{day}
    \clearfield{note}
}

\DeclareFieldFormat*{title}{%
  \iffieldundef{url}
    {%
      \iffieldundef{doi}
        {#1}%
        {\href{https://doi.org/\thefield{doi}}{#1}}%
    }%
    {\href{\thefield{url}}{#1}}%
}

\AtEveryBibitem{%
  \ifentrytype{book}{%
    \clearfield{labelyear}
    \renewbibmacro*{date}{%
      \printtext[parens]{\printfield{year}}%
    }%
  }{}%
}

\AtEveryBibitem{%
  \ifentrytype{incollection}{%
    \clearfield{labelyear}
    \renewbibmacro*{date}{%
      \printtext[parens]{\printfield{year}}%
    }%
  }{}%
}

\DeclareBibliographyDriver{misc}{%
  \usebibmacro{author/translator+others}%
  \setunit{\labelnamepunct}\newblock
  \printfield[title]{title}%
 \newunit\newblock
  \iffieldequalstr{eprinttype}{arxiv}{%
   \StrBefore{\thefield{eprint}}{ }[\arxivnumber]%
    \printtext{arXiv:\arxivnumber\space(\printfield{year})}%
    }{%
    \printfield{url}%
  }%
  \newunit\newblock
  \usebibmacro{finentry}%
}

\DeclareFieldFormat{pages}{%
  \ifboolexpr{
    test {\iffieldnums{pages}}
    and
    not test {\iffieldnum{pages}}
  }
    {pp.\,#1}
    {#1}%
}

\usepackage[capitalize]{cleveref}
\crefname{equation}{}{}

\newtheorem{lemma}{Lemma}
\newtheorem*{lemma*}{Lemma}
\newtheorem{theorem}[lemma]{Theorem}
\newtheorem*{theorem*}{Theorem}
\newtheorem{proposition}[lemma]{Proposition}

\theoremstyle{definition} 
\newtheorem*{remark}{Remark}
\title{A second order upper bound to the free energy of the two
dimensional Bose gas}

\author{
Florian Haberberger\thanks{Department of Mathematics, LMU Munich, Theresienstra\ss e 39, 80333 Munich, Germany. 
mail: haberberger@math.lmu.de, lukas.junge@math.lmu.de} ,
Lukas Junge\footnotemark[1]
}

\begin{document}
	\maketitle
    \begin{abstract}
We consider a two-dimensional Bose gas in the dilute regime where $\rho \fa^2$ is small. For temperatures below the Berezinskii--Kosterlitz--Thouless critical temperature, we derive an explicit upper bound for the free energy density using Bogoliubov theory. Our result captures the contribution of quasiparticle modes with dispersion relation $\sqrt{p^4 + 8\pi \rho\, \delta\, p^2}$ and where $\delta = 2 / (|\log(\rho \fa^2)| + \log |\log(\rho \fa^2)|)$.
\end{abstract}
\section{Introduction}
The thermodynamic properties of the ideal Bose gas have been understood since the pioneering works of Bose and Einstein \cite{Bose1924,Einstein1925}. In contrast, a rigorous derivation of interacting Bose gases from the many-body Schrödinger equation remains a central open problem in mathematical physics. In the dilute regime, interaction effects are weak but produce universal corrections to both the ground state energy and the excitation spectrum, depending only on the density $\rho$ and scattering length $\fa$.

At zero temperature $T=0$, the ground-state energy density $e(\rho)$ of a two-dimensional Bose gas exhibits a logarithmic dependence on the gas parameter $\rho \fa^2$, as first predicted by Schick in 1961 \cite{Schick}. 
Subsequent work has provided increasingly accurate expansions of the ground-state energy density \cite{Andersen2002, Cherny2D}. The expression given below is due to Mora and Castin \cite{Mora2D}.

\begin{equation*}\label{Lee huang yang in 2dimensons}
e(\rho)=2\pi\rho^2\delta\left(1+\Big(\gamma+\frac{1}{4}+\frac{\log (\pi)}{2}\Big)\delta +o (\delta) \right),
\qquad \delta=\frac{2}{\vert\log(\rho \fa^2)\vert+\log\vert \log(\rho \fa^2)\vert},
\end{equation*}
where $\gamma=0.577\ldots$ denotes the Euler--Mascheroni constant. 
The leading term was rigorously established by Lieb and Yngvason \cite{Lieb2001}, and the second-order correction was proved by Fournais et al. \cite{FGJMO-22}. 
In particular, these results confirm the universality of the dilute Bose gas in two dimensions, in the sense that up to second order the dependence on $v$ only appears through the scattering length $\fa$. 
Moreover, the Bogoliubov excitation spectrum has been rigorously studied in two dimensions. 
Using Bogoliubov theory it was established in \cite{caraciExcitationSpectrumTwodimensional2023} that in the Gross--Pitaevskii scale the low energy spectrum is given by
\begin{equation*}\label{excitation 2d}
    \sum_{p}n_p \sqrt{p^4+8\pi \delta \rho p^2},\quad n_p \in \mathbb{N}_0.
\end{equation*}
This suggests that for temperatures where the Bogoliubov theory applies, the free energy density is described by
\begin{equation}\label{free energy 2d}
    e(\rho) + \frac{T}{(2\pi)^2}\int_{\mathbb{R}^2}\log \left( 1-e^{-\frac{1}{T}\sqrt{p^4+8\pi \rho \delta p^2}} \right) \dd p.
\end{equation}
We show that this result, and in particular the Bogoliubov approximation, is valid up to temperatures comparable to the Berezinskii--Kosterlitz--Thouless critical temperature, see \cite{JMKosterlitz_1973,Berezinskii:1972fet}, 
\begin{equation}\label{Kosterlitz Thouless critical temperature}
    T_c=\frac{4\pi \rho}{\vert \log \delta\vert}.
\end{equation}
For the rigorous derivation of the first-order expansion at higher temperatures, see \cite{DEUCHERT_MAYER_SEIRINGER_2020,Mayer2020}.

It is somewhat surprising that \cref{free energy 2d} remains valid for $T>0$ despite the absence of Bose--Einstein condensation, cf. \cite{Mermin_Wagner_1966,Hohenberg1967}. This can be understood by noting that the free energy is a local quantity, even in the thermodynamic limit. Since the system exhibits Bose--Einstein condensation on sufficiently short length scales at positive temperature, the Bogoliubov prediction continues to provide an accurate description.

The analogous formula for the free energy density in three dimensions has been derived to Lee--Huang--Yang precision in \cite{haberberger2024freeenergydilutebose,haberbergerUpperBound2024}; see also \cite{FGJMOT-24} for a generalization to singular potentials. However, in three dimensions the temperature range is currently restricted to $T \sim \rho \fa \sim (\rho \fa^3)^{1/3} T_c^{\rm 3D}$. 
For results on the leading term at temperatures of order $T_c^{3D}$ we refer to \cite{Seiringer2008_FreeEnergy_3D_Lower, Yin2010_Free_UpperBound, Deuchert_2025_newupperboundfree}.
We emphasize that our result is not a straightforward adaptation of the three-dimensional case. The logarithmic behavior requires an analysis similar to the one needed in the hardcore setting. Furthermore, we obtain a formula that remains valid for temperatures comparable to the critical temperature, a regime that has not yet been accessible in the three-dimensional setting.

\subsection{Setting and Main Result}
	We consider the Hamiltonian
	\begin{equation}\label{hamiltonian}
		H_n=\sum_{i=1}^n-\Delta_i+\sum_{i<j}^n v(x_i-x_j)
	\end{equation}
acting on the symmetric $n$-particle Hilbert space $L_s^2(\Omega^n)$, where $\Omega\subset \R^2$ is a square with Dirichlet boundary conditions and $v\geq 0$ is the interaction potential.
The free energy at temperature $T\geq 0$ is given by
\begin{equation}\label{free energy}
	F_\Omega(n) = \inf\left\{\Tr(H_n\Gamma)-TS(\Gamma) \mid \Gamma \geq 0, \tr\Gamma = 1 \right\} = -T\log \tr e^{-H_n/T}.
\end{equation}
Here $S(\Gamma)=-\Tr(\Gamma\log\Gamma)$ is the von Neumann entropy.
Let $\rho >0$ denote the given particle density of the system.
The free energy density in the thermodynamic limit is defined as
\begin{equation}
    f(\rho,T) := \lim_{\substack{n,|\Omega| \to \infty \\ n/|\Omega| \to \rho}} \frac{F_\Omega(n)}{|\Omega|}.
\end{equation}
It is well established that the limit exists and that it is independent of the boundary conditions \cite{ruelle1969statistical, robinson_thermodynamic_1971}.
We denote by $\fa$ the scattering length of the potential $v$, see \Cref{sect:Scattering}.

\begin{theorem} \label{thm:main}
	Let $v$ be a positive, radial, compactly supported potential with scattering length $\fa>0$. There exists a constant $c>0$ depending only on the support of $v$ and its scattering length, such that for $T\leq c T_c$ and $\rho \fa^2\leq c$ we have
	\begin{align*}
		f(\rho,T) &\leq 2\pi\rho^2\delta\left(1+\Big(\gamma+\frac{1}{4}+\frac{\log (\pi)}{2}\Big)\delta\right) + \frac{T}{(2\pi)^2}\int_{\mathbb{R}^2} \log\left(1- e^{-\frac{1}{T}\sqrt{ p^4+8 \pi\rho \delta p^2}}\right) \dd p \\ &\quad+ C \rho^2\delta\left(\delta\vert\log(\delta)\vert+\frac{T}{T_c} \right)^2
	\end{align*}
    for some constant $C>0.$
	Here $\gamma=0.577\ldots$ is the Euler-Mascheroni constant and
	\begin{equation*}
		\delta = \frac{2}{\vert\log(\rho \fa^2)\vert+\log\big(\vert\log(\rho  \fa^2)|\big)}.
	\end{equation*}
\end{theorem}

\begin{remark}
\begin{enumerate}[label=\bf\arabic*.]
\item After a computation one finds that the ideal gas of density $\rho$ exhibits BEC in boxes of size $L$ and temperature $T$ if
    \begin{equation}\label{BEC inequality}
        \rho \geq T\log(T L^2).
    \end{equation}
This suggests that condensation occurs at the GP scale $L\sim \rho^{-\frac{1}{2}}\delta^{-\frac{1}{2}}$ for $T\leq T_c$. In this way the Berezinskii--Kosterlitz--Thouless critical temperature arises naturally in our analysis.

\item The diluteness parameter $\delta$ is not the same in the literature. It agrees with the first two terms of the diluteness parameter in \cite{Cherny2D} (“$u$") and in \cite{Mora_2003} (“$\epsilon(\rho)$"), which is accurate enough for our level of precision.
The diluteness parameter used in \cite{Lieb2001,GreenBook} is
\begin{equation} \label{eq:def_Y}
    Y:= |\log (\rho\fa^2)|^{-1}.
\end{equation}
Note that 
$$
\delta = \frac{2Y}{1+Y|\log Y|} = 2Y \left(1 - Y |\log Y| +   \frac{Y^2|\log Y|^2}{1+Y|\log Y|} \right) \leq 2Y.
$$
Theorem~\ref{thm:main} then reads
\begin{align*}
f(\rho,T) &\leq 4\pi\rho^2Y\left(1-Y\vert \log Y\vert+ \Big(2\gamma+\frac{1}{2}+\log (\pi)\Big)Y\right) \\&+\frac{T}{(2\pi)^2}\int_{\mathbb{R}^2} \log\left(1- e^{-\frac{1}{T}\sqrt{ p^4+16 \pi\rho Yp^2}}\right) \dd p+ C \rho^2Y\left(Y\vert\log(Y)\vert+\frac{T}{T_c} \right)^2,
\end{align*}
where we used the monotonicity of the thermal term.

    \item Theorem~\ref{thm:main} yields a better precision than \cite{Mayer2020} for $T/T_c \leq \sqrt{\delta|\log \delta|}$. In this regime, \cite{Mayer2020} gives
    $$
    f(\rho,T) \leq \frac{T^2}{(2\pi)^2}\int_{\R^2}\log(1-e^{-p^2})\dd p + 4\pi \rho^2Y(1+2 T/T_c) + C\rho^2 Y^2|\log Y|
    $$
    where we used that for $T\ll T_c$ the free energy density of the ideal gas satisfies 
    $$
    f_0(\rho,T) = \frac{T^2}{(2\pi)^2}\int_{\R^2}\log(1-e^{-p^2})\dd p + \mathcal{O}\left(\rho^2Y \frac{1}{|\log Y|^2}(T/T_c)^2\right).$$
    Expanding the square root in \Cref{thm:main} yields
    \begin{align*}
    & \frac{T}{(2\pi)^2}\int_{\R^2}\log\left(1-\exp^{-\frac{1}{T}\sqrt{p^4+8\pi\rho\delta p^2}}\right) \dd p 
    \leq \frac{T}{(2\pi)^2}\int_{\R^2}\log\left(1-\exp^{-\frac{1}{T}(p^2+4\pi\rho\delta)}\right) \dd p 
    \\
    &= -\frac{T^2}{4\pi} {\rm Li}_2(e^{-4\pi\rho\delta/T})
    \\
    &\leq \frac{T^2}{(2\pi)^2}\int_{\R^2}\log(1-e^{-p^2})\dd p - T\rho\delta \log(4\pi\rho\delta/T) + T\rho\delta + C T^2 (\rho \delta/T)^2
    \\
    &\leq \frac{T^2}{(2\pi)^2} \int_{\R^2}\log(1-e^{-p^2})\dd p + 4\pi\rho^2 \delta T/T_c  + C \rho^2\delta^2,
    \end{align*}
    where we used $T < T_c$ in the last inequality and where ${\rm Li}_2(z) = \sum_{k=1}^\infty \frac{z^k}{k^2}$ denotes the dilogarithm.
    For higher temperatures, we expect, similarly to the results of Mayer \cite{Mayer2020}, the appearance of a contribution of order $\frac{T^2}{T_c^2}$. We believe that this term can be derived by accounting for quartic correlations.
\end{enumerate}
\end{remark}

\subsection{Trial state and proof idea}
Theorem~\ref{thm:main} follows from the variational principle for the free energy \cref{free energy} by constructing a suitable trial state. Our construction is inspired by \cite{FGJMO-22}, but in contrast to that work, we must incorporate excitations. In particular, we need to keep track of contributions from terms in the Hamiltonian that would vanish on the vacuum in the setting of \cite{FGJMO-22}.

We employ a standard localization technique, namely we partition the thermodynamic box into squares of fixed side length $L$, equipped with periodic boundary conditions. On each square, we construct a grand canonical trial state with the desired energy; this is established in \Cref{lem:trial_hamil} and \Cref{lem:trial_entropy_easy}. 
Achieving this at positive temperature and is the main novelty of this paper and the main part of the analysis.
The content of \Cref{prop:Localization} is that this grand canonical trial state can be extended to a canonical trial state with Dirichlet boundary conditions. Moreover, following \cite{robinson_thermodynamic_1971}, these localized Dirichlet states can be patched together to produce a trial state on the full thermodynamic domain $\Omega$.

We choose the side length of the small boxes to be
\begin{equation}\label{eq. size of L}
    L =  \rho^{-1/2}Y^{-\alpha}
\end{equation}
for a parameter $\alpha>0$, which will eventually be fixed to $\alpha = \frac{5}{2}$. According to \cref{BEC inequality}, Bose--Einstein condensation is expected to occur at this length scale.

We denote by $a_p=a\big(\frac{1}{L}e^{ip\,\cdot}\big)$ for $p\in\Lambda^* := \frac{2\pi}{L}\mathbb{Z}^2$ the usual annihilation operator and analogously the creation operator $a_p^*$. They satisfy the canonical commutation relations
$$
[a_p,a_q^*] = \delta_{p,q},\quad [a_p,a_q] = 0 = [a_p^*,a_q^*],
$$
on the Fock space $\cF := \bigoplus_{n=0}^\infty L^2_s(\Lambda^n)$. 
Moreover, we denote the set of excited momenta by $\Lambda_+^* := \Lambda^*\setminus\{0\}$ and the excited Fock space by $\cF_+ := \mathcal{F}(\langle \1_\Lambda\rangle^\perp)$. 
The number operator is denoted by $\cN.$ $\cN_+$ is the number operator on $\cF_+$ and $\cN_0 = a_0^*a_0 = \cN- \cN_+.$
The grand-canonical Hamiltonian on the box $\Lambda$ is given by
\begin{equation} \label{eq:def_cH}
    \cH = \bigoplus_{n=0}^\infty H_n = \sum_{p\in\Lambda^*}p^2a_p^*a_p + \frac{1}{2|\Lambda|} \sum_{p,q,k \in \Lambda^*} \hat{v}(k) a_p^*a_{q+k}^*a_q a_{p+k}.
\end{equation}
Here we use the following definition of the Fourier transform
$$
\hat{f}(p) = \int_{\R^2}f(x) e^{-ipx}\dd x,\quad f(x) = \int_{\R^2} \hat{f}(p)e^{ipx}\frac{\dd p}{(2\pi)^2}.
$$

As a first step, we soften the potential by conjugating $\cH$ with a Jastrow factor $F$. This allows us to replace $v$ by a more regular potential $\tilde{v}$ which has the same scattering length but small $L^1$-norm, see \Cref{sect:softening}. 
The softened Hamiltonian can be approximately diagonalized via two unitary transformations, namely a quadratic Bogoliubov transformation $e^{\cB}$ that implements correlations and a Weyl transformation $W_{N_0(1+\eta)}$ that introduces a condensate of $N_0(1+\eta)$ particles in the trial state, see \Cref{sect:trial_state}. Here $N_0 \in [0,\rho L^2]$ is the expected number of particles in the condensate and $\eta \in [-\frac{1}{2},\frac{1}{2}]$ is a small parameter that ensures that the trial state has the correct density in the end.
The resulting quadratic part of the Hamiltonian is of the form
\begin{equation}\label{definition of HD}
    H_D := \sum_{p \in \Lambda_+^*} \sqrt{p^4+8\pi \delta \rho p^2} a_p^*a_p
\end{equation}
in agreement with \cref{excitation 2d}. 
We denote the Gibbs state of $H_D$ as
\begin{equation*}
    \Gamma_D=\frac{e^{-\beta H_D}\1_{\{\mathcal{N}_0=0\}}}{\Tr_{\mathcal{F}_+}(e^{-\beta H_D})}.
\end{equation*}
The trial state is defined as follows
\begin{equation} \label{eq:def_Gamma_0}
\Gamma_0 = \frac{F \weyleta e^{\cB} \Gamma_D e^{-\cB} \weyleta^* F}{\tr\left(F \weyleta e^{\cB} \Gamma_D e^{-\cB} \weyleta^* F\right)}.
\end{equation}
Moreover, we introduce the length scale $b$ up to which we implement our correlations with the Jastrow factor, see \cref{Jastrow factor} and \cref{Def: of fb} for the exact role of $b$. We will always consider 
\begin{equation}\label{equation smallness of b}
    R\leq b\leq \rho L^{-1}Y 
\end{equation}
with $R$ such that $\supp(v) \subset B(0,R)$, i.e., we only implement correlations at distances less than the average particle distance. In the end we minimize over $b$ and find $b=\rho L^{-1}Y^3$.

We conclude this section by stating the main lemmata and deducing \Cref{thm:main} from them.
\begin{lemma} \label{lem:trial_hamil}
Let $v, T$ be as in \Cref{thm:main},  $\alpha > 3/2$. Then for any $\tilde{\rho}\in [\rho,(1+Y^2)\rho]$ there exists $\eta\in[-\frac{1}{2},\frac{1}{2}]$ such that 
\begin{align} \label{eq:N_Gamma}
 \tr(\cN \Gamma_0) =  L^2  \tilde{\rho}.
 \end{align}
Furthermore, $\Gamma_0$ satisfies
 \begin{align*}
 L^{-2} \tr(\mathcal{H}_v \Gamma_0) &\leq 2\pi \rho^2 \delta\left(1+\delta\Big(\gamma + \frac{1}{4}+\frac{\log(\pi)}{2}\Big) \right) + L^{-2} \Tr(H_D\Gamma_D) 
 \\& + C \rho^2 \frac{1}{\log\left(\frac{b}{\fa}\right)} \left(1-\delta \log\Big(\frac{b}{\fa}\Big)+Y+\frac{T}{T_c}\right)^2 + C b^2 \rho^4 L^2
 \end{align*}
for some constant $C$ only depending on $\alpha$.
 \end{lemma}

The second Lemma below states that the entropy is almost unchanged by the Jastrow factor i.e $S(\Gamma_0) \approx S(\weyl e^{\cB} \Gamma_D e^{-\cB} \weyl^*) = S(\Gamma_D)$.

\begin{lemma} \label{lem:trial_entropy_easy}
For $T$ as in Theorem~\ref{thm:main} we have
$$
-L^{-2}TS(\Gamma_0) 
\leq -L^{-2}TS(\Gamma_D) + C \rho^3 bL \vert\log(b\rho L)\vert
$$
for some constant $C>0.$
\end{lemma}

We can extend the trial state $\Gamma_0$ to the thermodynamic box $\Omega$ by first obtaining a Dirichlet version of $\Gamma_0$ on a slightly smaller box and then gluing copies of them together. 

\begin{proposition}\label{prop:Localization}
Let $0<R<\ell<L$ such that $\supp(v)\subset B(0,R).$ Let $\Gamma_L$ be a normalized density matrix on the Fock space $\cF(\Lambda_L)$, satisfying periodic boundary conditions. 
If
\begin{equation} \label{eq:enough_part}
    \rho = \frac{\tr(\cN \Gamma_L)}{(L+2\ell+R)^2},
\end{equation}
then there exists a universal constant $C$ such that
$$
f({\rho},T) \leq \frac{\tr(\cH\Gamma_L) - TS(\Gamma_L)}{(L+2\ell+R)^2} + \frac{C}{L^3\ell}{\tr(\cN \Gamma_L)}.
$$
\end{proposition}

We finish this section by combining the lemmata and proving Theorem~\ref{thm:main}.

\begin{proof}[Proof of \Cref{thm:main}]
Let $\ell=\frac{1}{5} Y^2 L$ in Proposition~\ref{prop:Localization}. Then, for $Y$ sufficiently small, Lemma~\ref{lem:trial_hamil} ensures that one can choose $\eta$ such that \cref{eq:enough_part} holds for $\Gamma_0$.
Thus \Cref{prop:Localization} implies
\begin{align} \label{eq:main_f_apriori}
f(\rho,T) &\leq \frac{\tr(\cH\Gamma_0) - TS(\Gamma_0)}{(L+2\ell+R)^2} +C\frac{{\rho}}{L\ell}.
\end{align}
We combine the thermal contributions of \Cref{lem:trial_hamil} and \Cref{lem:trial_entropy_easy} and obtain
\begin{align}\label{eq:thermal}
\begin{split}
    &L^{-2} \Tr(H_D\Gamma_D) -L^{-2}TS(\Gamma_D) =TL^{-2} \sum_{p\in\Lambda_+^*} \log\left(1-e^{- \tfrac{1}{T}\sqrt{p^4+8\pi\rho \delta p^2
    }}\right)  
    \\
    &\leq \frac{T}{(2\pi)^2} \int_{\R^2} \log\left(1-e^{-\tfrac{1}{T}\sqrt{p^4+8 \pi p^2\rho\delta}}\right) \dd p + CL^{-1} T^{3/2}.
    \end{split}
\end{align}
Thus using \Cref{lem:trial_hamil} and \Cref{lem:trial_entropy_easy}  in \cref{eq:main_f_apriori} and applying \eqref{eq:thermal} yields
\begin{align*}
    f(\rho,T) 
    &\leq 2\pi \rho^2 \delta\left(1+\delta\Big(\gamma + \frac{1}{4}+\frac{\log(\pi)}{2}\Big) \right) + \frac{T}{(2\pi)^2} \int_{\R^2} \log\left(1-e^{-\tfrac{1}{T}\sqrt{p^4+8 \pi p^2\rho\delta}}\right) \dd p 
    \\& + C \rho^2 \frac{1}{\log\left(\frac{b}{\fa}\right)} \left(1-\delta \log\Big(\frac{b}{\fa}\Big)+Y+\frac{T}{T_c}\right)^2 + C b^2 \rho^4 L^2  + C \rho^3 bL |\log(\rho bL)| 
    \\&+ C (\rho^2 Y + T^2) \frac{\ell}{L} + C\frac{{\rho}}{L\ell}+C L^{-1} T^{3/2}.
\end{align*}
Recalling $\ell=\frac{1}{5}Y^2L$, we choose $L=\rho^{-\frac{1}{2}}Y^{-\frac{5}{2}}$ and $b=\rho^{-\frac{1}{2}}Y^{\frac{11}{2}}$. This yields the result by noting $0 \leq 1- \delta \log(b/\fa) \leq C Y |\log (Y)|$ and $\log(b/\fa)^{-1}\leq C Y$.
\end{proof}

\section{Preliminaries}
\subsection{Scattering Length}\label{sect:Scattering}
\paragraph{Definition}
	Let $v:\R^2\to\R$ be positive, radial and compactly supported. We define its scattering length $\fa(v)$ by the variational problem
	\begin{equation} \label{def:scattering length}
		\frac{2\pi}{\log(\frac{R}{\fa(v)})}=\inf\left\{\int_{B(0,R)}\vert \nabla\phi\vert^2+\frac{1}{2}v|\phi|^2\dd x \;\mid\; \phi\in H^1(B(0,R)),  \phi_{\vert_{\partial B(0,R)}}=1\right\}
	\end{equation}
for $R$ such that $\rm{supp}(v)\subset B(0,R)$. 

It is well known that the variational problem \eqref{def:scattering length} has a unique minimizer $\phi_R$, see \cite[Appendix C]{GreenBook}. Moreover, $0\leq \phi_R\leq 1$, $\phi_R$ is radial and satisfies the scattering equation in $B(0,R)$
\begin{equation}\label{eq:scattering_equation}
	(-\Delta +\frac{1}{2}v)\phi=0.
\end{equation}
Let ${\rm supp}(v) \subset B(0,R_0)\subset B(0,R)$. Then from the fundamental solution of the Laplace equation in two dimensions we obtain
\begin{equation}\label{explicit formula for scattering solution}
	\phi_R(x) = \frac{\log(|x|/\fa)}{\log(R/\fa)}, \; R_0 < |x| < R.
\end{equation}
By comparing with the solution in $B(0,R_0)$, it is clear that $\fa$ is independent of $R$.

\subsection{Jastrow Factor}
In our trial state we will use a Jastrow factor $F:\cF\to\cF$. It was introduced in this form in \cite{Jastrow1955} and for example used by Dyson \cite{Dyson_Hard_1957} in his proof of an upper bound for the ground state energy of a Dilute Bose gas in 3D. It was also used with great success in the study of the 3D ground state \cite{Basti_2023_Hard_GP} and \cite{BCGOPS_Hard_2024} as well as in 2D \cite{FGJMO-22}.
The idea is to implement correlations on short distances, which results in a softening of the potential.
The action of $F$ on the $n$-particle sector $L_s^2(\Lambda^n)$ is given by
\begin{equation}\label{Jastrow factor}
	F_n(x)=\prod_{i<j}^n f(x_i-x_j),\qquad x\in \mathbb{R}^{2 n},
\end{equation}
where $f$ is a truncated scattering solution
\begin{equation}\label{Def: of fb}
f(x) = 
\begin{cases}
    \phi_b(x), & |x|<b \\
    1, &|x|\geq b.
\end{cases}
\end{equation}
Here $\phi_b$ is the scattering solution \eqref{explicit formula for scattering solution} that is normalized to $1$ at $|x| = b$.
In the end we choose $b$ to be much bigger than the scattering length $\fa$ however much smaller than the average particle distance $\rho^{-1/2}$. When we conjugate the Hamiltonian with $F$ we obtain an effective potential $\tilde{v}$ whose properties we summarize in the following lemma.
\begin{lemma}\label{properties of sof potential}
	Let $v\colon \R^2\to\R$ be positive, radial and compactly supported with scattering length $\fa$, and let $f$ be its associated scattering solution truncated at $b$ such that $\supp(v)\subset B(0,b)$. Then
	\begin{equation} \label{eq:def_v_tilde}
	    (-2\Delta+v)f=\delta_{\vert x\vert =b}\frac{2}{b\log(\frac{b}{\fa})}=:\tilde v,
	\end{equation}
    in a distributional sense.
	Furthermore, the scattering lengths of $v$ and $\tilde{v}$ agree. Lastly, we have the inequality 
\begin{equation}\label{eq: difference between v and g}
    \hat{\tilde{v}}(0) = \frac{4\pi }{\log(\frac{b}{\fa})} ,\qquad \vert\widehat{\tilde{v}}(p)\vert \leq  C\frac{\widehat{\tilde{v}}(0)}{\sqrt{b\vert p\vert}}.
\end{equation}
\end{lemma}
\begin{proof}
Equation \eqref{eq:def_v_tilde} and agreement of the scattering lengths are straightforward computations, see also \cite[Lemma 3.10]{FGJMO-22}. 
Since $\tilde{v}$ is a uniform measure on a sphere, its Fourier transform is given by the Bessel function $J_0$, from which the decay properties follow.
\end{proof}

After applying the Jastrow factor, we obtain a renormalized Hamiltonian that we aim to diagonalize. To achieve this, we take the scattering solution $\tilde{\phi}_{\tilde{R}}$ of $\tilde{v}$ on $B(\tilde{R},0).$ Here
\begin{align*}
	\tilde{R}/\fa &:= \frac{1}{\sqrt{\rho \fa^2 Y}} \gg 1,
\end{align*}
so that $ \delta = \log(\tilde{R}/\fa)^{-1}.$
Let $\tilde{\phi}$ be an extension of $\tilde\phi_{\tilde{R}}$ to the full space, that is
$$
\tilde{\phi}(x) = \begin{cases}
    \phi_{\tilde{R}}(x), &|x|<\tilde{R} \\
    \delta \log(|x|/\fa), & |x|\geq \tilde{R}. \end{cases}
    $$
In particular $\tilde{\phi}$ satisfies the scattering equation \eqref{eq:scattering_equation} on $\R^2$.
We further define $\omega := 1-\tilde\phi$ and 
\begin{equation} g := \tilde\phi \tilde v. \label{def: of g}\end{equation} Note that we do not put tildes over $\omega$ and $g$ due to simplicity, but they always refer to the softened potential $\tilde{v}$.

In the following lemma we collect some facts about $g$, see also \cite[Chapter 3.4]{FGJMO-22}. 
\begin{lemma}
    It holds that
\begin{align}
	\hat{g}(0) &= 4\pi\delta \label{eq:g_hat_0}
\intertext{and} 
\label{eq:gw_hat}
\hat{g\omega}(0) 
&= \int_{\R^2} \frac{\hat{g}(p)^2-\hat{g}(0)^2\1_{\{|p|\leq 2e^{-\gamma-1/\delta}\fa^{-1}\}}}{2p^2} \frac{\dd p}{(2\pi)^2}.
\end{align}
Furthermore,
\begin{equation} \label{eq:g_decay}
\vert \widehat{g}(p)\vert \leq c_1\frac{\widehat{g}(0)}{\sqrt{R\vert p\vert}} \;\text{ for }\; \vert p\vert \geq \fa^{-1}.
\end{equation}
\end{lemma}

\begin{proof}
We compute with \eqref{eq:scattering_equation} and partial integration 
$$
\hat{g}(0) = 2\int_{B(0,\tilde{R})} \Delta \phi_{\tilde{R}} = 2\int_{\partial B(0,\tilde{R})} \nabla\phi_{\tilde{R}}\cdot \vec{n}\dd S = 4\pi \delta.
$$
From \eqref{eq:scattering_equation} we find $-\Delta \omega = \frac{g}{2}$. Thus $\hat{\omega}(u)=\int_{\R^2}\frac{\hat{g}(p) }{2p^2}u(p)\dd p$ for all test functions $u\in C_c^\infty(\R^2\setminus\{0\}).$ We write
$\omega = -\delta\log(e^{-1/\delta}|x|/\fa)+\tilde{\omega}$ so that $\tilde{\omega}\in L^1(\R^2)\cap L^\infty(\R^2).$ Let $u\in C_c^\infty(\R^2)$ be a test function. Then for all $\vep>0$ we may write $u = u_\vep + (u-u_{\vep})$ for some $u_\vep \in C_c^\infty(B(0,\vep))$ such that $u-u_{\vep}\in C_c^\infty(\R^2\setminus\{0\}).$ We obtain
\begin{align*}
\braket{\hat{\omega},u} = \hat{\tilde{\omega}}(u_\vep) - \delta \braket{\reallywidehat{\log(e^{-1/\delta}|x|/\fa)}, u_\vep} + \int_{\R^2} \frac{\hat{g}(p)}{2p^2}(u(p)-u_{\vep}(p)) \dd p.
\end{align*}
We immediately get $\hat{\tilde{\omega}}(u_\vep) \xrightarrow{\vep \to 0} 0$. Moreover, one can calculate the distributional Fourier transformation of the logarithm, see e.g.\ \cite[9.8d)]{vladimirov1971equations}, and for $u\in C_c^\infty(\R^2)$ one obtains
$$
-\delta \braket{\reallywidehat{\log(e^{-1/\delta}|x|/\fa)},u} = 4\pi \delta \int_{\R^2}\frac{u(p)-u(0)\1_{\{|p|\leq2e^{-\gamma}e^{-1/\delta}\fa^{-1}\}}}{2p^2}\dd p.
$$
Since $\hat{g}(0) = 4\pi\delta$ this implies
\begin{align*}
    \braket{\hat{\omega},u} &= \int_{\R^2} \frac{\hat{g}(p)u(p)-\hat{g}(0)u(0)\1_{\{|p|\leq2e^{-\gamma-1/\delta}\fa^{-1}\}}}{2p^2}\dd p + \lim_{\vep \to 0} \int_{\R^2}\frac{u_\vep(p)}{2p^2}(\hat{g}(0)-\hat{g}(p))\dd p
    \\
    &= \int_{\R^2} \frac{\hat{g}(p)u(p)-\hat{g}(0)u(0)\1_{\{|p|\leq2e^{-\gamma-1/\delta}\fa^{-1}\}}}{2p^2}\dd p,
\end{align*}
where we used that $u_\vep \to 0$ almost everyhwere and that $\frac{\hat{g}(0)-\hat{g}(p)}{p^2}$ is locally integrable. This implies \eqref{eq:gw_hat} by noting that $\hat{g\omega}(0) = \braket{\omega,g}=\frac{1}{(2\pi)^2}\braket{\hat{\omega},\hat{g}}$. The decay property \eqref{eq:g_decay} follows as $g$, like $\tilde{v}$, is a uniform surface measure on a sphere.
\end{proof}

\section{Proof of Lemma~\ref{lem:trial_hamil}}

\subsection{Trial state}\label{sect:trial_state}
We start by providing the exact definition of the unitaries in the trial state \cref{eq:def_Gamma_0}. The Weyl transformation $W_{\vep}$ generates a condensate with on average $\vep>0$ particles
\begin{equation}\label{eq:weyl}
W_\vep := {\rm exp}[\sqrt{N_0} (a_0^*-a_0)]
\quad \text{ so that } \quad
	W_\vep^* a_p W_\vep = \delta_{p,0}\sqrt{\vep}+a_p.
\end{equation}
It is straightforward to show that this is a strongly continuous unitary group, in particular $W_{N_0(1+\eta)} = W_{N_0\eta}W_{N_0}$. Therefore, we will focus our analysis first on the state
$\weyl e^{\cB} \Gamma_D e^{-\cB} \weyl^*$ and consider the effect of the additional Weyl transformation in the end of the section.
Here $N_0$ is the average the number of particles in the condensate and will be defined in \eqref{Choice of N_0}.
Next, $e^\mathcal{B}$ is a quadratic Bogoliubov transformation. It is given by
$$
\cB = \frac{1}{2} \sum_{p\in \Lambda_+^*} \vphi_p a_p^*a_{-p}^* - \hc
$$ 
We adopt the kernel $\vphi$ as in \cite{FGJMO-22}, originally introduced in three dimensions in \cite{Schlein_2008_UpperSoft}. This choice is motivated by minimizing the energy of the Bogoliubov-rotated vacuum.
$$
\tanh(2\vphi_p) = -\frac{\rho_0 \hat{g}(p)}{p^2+\rho_0\hat{g}(p)}.
$$
For $p\neq 0$ we find
\begin{equation}\label{eq. bogolibov application}
\begin{aligned}
    e^{-\cB} a_p e^{\cB} &= a_p\sqrt{\frac{1}{2}\left(\frac{p^2+\rho_0\hat{g}(p)}{D_p}+1\right)} - a_{-p}^*\sqrt{\frac{1}{2}\left(\frac{p^2+\rho_0\hat{g}(p)}{D_p}-1\right)} \\ &=: c_{p} a_{p} + s_{p} a_{-p}^*,
\end{aligned}
\end{equation}
with $\rho_0 := \frac{N_0}{\vert \Lambda\vert}$. 
Here $D_p := \sqrt{p^4+2p^2\rho_0\hat{g}(p)}$ is well defined due to $\hat{g}(0) > 0$ and the continuity of $\hat{g}$.

Moreover, this choice will turn out to naturally satisfy $c_ps_p \approx -\rho_0 \hat\omega(p)$, see \cite[Section 4]{solovej2025mathematicalphysicsdilutebose} for a nice heuristic argument why this is desired.

Lastly we recall $H_D = \sum_{p\in\Lambda_+^*}\sqrt{p^4+8\pi \rho \delta p^2}a_p^*a_p$ and its Gibbs state
\begin{equation*}
    \Gamma_D=\frac{e^{-\beta H_D}\1_{\{\mathcal{N}_0=0\}}}{\Tr_{\mathcal{F}_+}(e^{-\beta H_D})}.
\end{equation*}

From \eqref{eq:weyl}, the fact that $\Gamma_D$ satisfies Wicks theorem and does not contain any particles in the zero mode we find 
\begin{align} \label{eq. this makes n0 possible} \begin{split}
    \tr(\cN \weyl e^{\mathcal{B}}\Gamma_D e^{-\mathcal{B}}\weyl^*) &= N_0 + \tr(e^{-\cB} \cN_+ e^\cB \Gamma_D)
\\
&= N_0+\sum_{p\in \Lambda^*_+}s_p^2 + \sum_{p\in \Lambda_+^*}(s_p^2+c_p^2) \Tr(a_p^*a_p \Gamma_D).
    \end{split}
\end{align}
We recall that $\cN_+ = \sum_{p\in\Lambda_+^*}a_p^*a_p$ is the number operator on the excitation Fock space. 
We choose $N_0$ such that 
\begin{equation}\label{Choice of N_0}
\tr(\cN\weyl e^{\mathcal{B}}\Gamma_D e^{-\mathcal{B}}\weyl^*) =\rho L^2=: N,
\end{equation}
where the existence of a solution to the above fixed point equation, bear in mind that $\cB$ depends on $N_0$, is due to the continuity in $N_0$ and the following proposition.
\begin{proposition}\label{Lem. Usefull bounds}
    Let $c_p$ and $s_p$ be given as in \cref{eq. bogolibov application}. Then for $\rho_0\leq \rho$ there exists a $C>0$ such that
    \begin{equation}\label{bound on s^2}
\sum_{p\in \Lambda^*_+}s_p^2 \leq CNY
    \end{equation}
    and 
    \begin{equation}\label{eq:sc_bound}
 \sum_{p\in \Lambda^*_+}\vert s_pc_p\vert \leq CN.
\end{equation}
For $T\leq \rho$ we have
    \begin{equation}\label{bound on Q2}
 \sum_{p\in \Lambda_+^*}(s_p^2+c_p^2) \Tr (a_p^*a_p \Gamma_D)\leq CN \frac{T}{T_c}.
    \end{equation}
    In particular, if $ C(\frac{T}{T_c} + Y) \leq \frac{1}{2}$, then \cref{eq. this makes n0 possible} admits a solution with \(\rho_0 \geq \frac{\rho}{2}\).
\end{proposition}

\begin{remark}
    The above proposition is in accordance with \cref{BEC inequality}, in the sense that we do not expect BEC at our length scale for $T\geq T_c$.
\end{remark}

\begin{proof}
To show \eqref{bound on s^2} we split the sum as follows,
\begin{align*}
    \sum_{p\in\Lambda_+^*}s_p^2 &\leq \left( \sum_{p^2\leq\rho_0 \hat{g}(0)} + \sum_{p^2>\rho_0 \hat{g}(0)} \right) \frac{p^2+\rho_0\hat{g}(p)-D_p}{2D_p} \\
    &\leq \sum_{p^2\leq\rho_0 \hat{g}(0)} \left[\frac{p^2+\rho_0\hat{g}(p)}{2|p|\sqrt{2\rho_0\hat{g}(p)}}\right] + C\sum_{p^2>\rho_0 \hat{g}(0)}\left[1+\rho_0\hat{g}(p)p^{-2}-\sqrt{1+2\rho_0\hat{g}(p)p^{-2}}\right]
    \\
    &\leq C L^2\rho\hat{g}(0),
\end{align*}
where we used the elementary inequality $1+x-\sqrt{1+2x}\leq x^2$ for $x \geq -(\sqrt{2}-1)$ and $\rho_0\leq \rho$. \Cref{eq:sc_bound} can be shown similarly, using the decay of $\hat{g}$ \eqref{eq:g_decay}.

Next, we show the bound \eqref{bound on Q2}. From the identities 
$$\tr (a_p^*a_p \Gamma_D) = \frac{1}{e^{\beta \sqrt{p^4+2\rho p^2\widehat{g}(0)}}-1} \quad \text{and} \quad 
s_p^2+c_p^2= (p^2+\rho_0\widehat{g}(p))/D_p
$$ 
we find
\begin{equation*}
  \sum_{p\in \Lambda_+^*}(s_p^2+c_p^2) \Tr(a_p^*a_p \Gamma_D) = \sum_{p\in \Lambda_+^*}\frac{p^2+\rho_0\widehat{g}(p)}{D_p}\frac{1}{e^{\beta \sqrt{p^4+2\rho p^2\widehat{g}(0)}}-1}.
\end{equation*}
We split the sum similar as before.
For $p^2 \leq \rho \hat{g}(0)$ we find 
\begin{align*}
& \sum_{p^2 \leq \rho \hat{g}(0)}  \frac{p^2+\rho_0\hat{g}(p)}{D_p} \frac{1}{e^{\beta \sqrt{p^4+2\rho p^2\widehat{g}(0)}}-1} 
\leq C \sum_{p^2 \leq \rho \hat{g}(0)} \left[\frac{\sqrt{\rho\hat{g}(0)}}{|p|}\frac{1}{e^{\beta |p|\sqrt{2\rho\hat{g}(0)}}-1}\right]
\\
& \leq 
C T \sum_{p^2 \leq \rho \hat{g}(0)} p^{-2} \leq CL^2T\int_{2\pi L^{-1}}^{\sqrt{\rho\hat{g}(0)}} |p|^{-1} \dd p
\leq C L^2 T \log(L^2 \rho\hat{g}(0)).
\end{align*}
For $p^2> \rho \hat{g}(0)$ we use $p^2/2\leq D_p$ and obtain
\begin{align*}
&\sum_{p^2 > \rho \hat{g}(0)}  \frac{p^2+\rho_0\hat{g}(p)}{D_p} \frac{1}{e^{\beta \sqrt{p^4+2\rho p^2\widehat{g}(0)}}-1} 
\leq 4\sum_{p^2 > \rho \hat{g}(0)} \frac{1}{e^{\beta p^2/2}-1} 
\leq CL^2 T \int_{\sqrt{\beta \rho \hat{g}(0)}}^\infty \frac{|p|}{e^{p^2/2}-1}\dd p 
\\
& \qquad = CL^2 T |\log(1-e^{-\beta \rho\hat{g}(0)/2})|
\leq CL^2 T \max\{|\log(\beta \rho \hat{g}(0)/2)|, 1\}.
\end{align*}
Combining these two inequalities and using our choice of $L$, \cref{eq. size of L} and the condition $T\leq \rho$ concludes the proof.
\end{proof}
To simplify notation we introduce
\begin{equation}\label{eq:def_Gamma_B}
    \Gamma_B=\weyl e^{\mathcal{B}}\Gamma_D e^{-\mathcal{B}}\weyl^*.
\end{equation}
Before we prove that $F$ softens the potential, we establish some a priori estimates for the expectation of the interacting part of the Hamiltonian in the state $\Gamma_B$. Since these estimates will be used repeatedly, we present them in generality.
\begin{lemma} \label{Lem:Hamil_energy_bound_general}
Let $v \geq 0$ be a compactly supported, radial potential and
\begin{equation} \label{eq:def_Hint}
\cH_v^{\rm int} = \frac{1}{2|\Lambda|}\sum_{p,q,k \in \Lambda^*} \hat{v}(k)a_p^*a_{q+k}^*a_qa_{p+k}.
\end{equation}
Then for $T$ as in \Cref{Lem. Usefull bounds} we have 
\begin{align}
  \Tr(\cH_v^{\rm int} \Gamma_B) \leq C \frac{\hat{v}(0)}{|\Lambda|}\left(N + \sum_{p\in\Lambda_+^*} (s_p^2+c_p^2)\tr(a_p^*a_p \Gamma_D) \right)^2 \leq C \hat{v}(0)  \rho N.
\end{align}
\end{lemma}

\begin{proof} 
From \cref{eq:def_cH} and \cref{eq:weyl} we obtain
    \begin{align} \label{eq:H_weyl}
\weyl^* \cH_v^{\rm int} \weyl&= \frac{N_0^2}{2|\Lambda|}\hat{v}(0) + \frac{N_0^{3/2}}{2|\Lambda|}\hat{v}(0)(a_0^*+a_0) + \frac{N_0}{|\Lambda|} \sum_{p \in \Lambda^*} \big(\hat{v}(p)+ \hat{v}(0)\big) a_p^*a_p \nn
\\
&   +\frac{N_0}{2|\Lambda|} \sum_{p \in \Lambda^*} \hat{v}(p) a_p^*a_{-p}^* + \hc  + \frac{N_0^{1/2}}{|\Lambda|}\sum_{p,q \in \Lambda^*} \hat{v}(p) a_{p+q}^*a_pa_q + \hc 
\\
& + \frac{1}{2|\Lambda|}\sum_{p,q,k \in \Lambda^*} \hat{v}(k)a_p^*a_{q+k}^*a_qa_{p+k}. \nn
\end{align}
Since $\Gamma_D$ and the Bogoliubov transformation $e^{\cB}$ are quadratic, the odd terms, i.e., the third term in the first line and the last term in the second line vanish in the trace. Moreover, $\Gamma_D$ does not have any particles in the zero mode so that all sums restrict to $\Lambda_+^*.$

Using $N_0 = N - \tr(e^{-\cB}\cN_+e^\cB\Gamma_D)$ from \cref{Choice of N_0} allows us to combine the first term and the second part of the third term as follows
\begin{equation}
    \Tr\left(\Big(\frac{N_0^2}{2\vert \Lambda\vert}+\frac{N_0}{\vert \Lambda\vert}\sum_{p\in\Lambda_+^*} a_p^*a_p\Big)e^{-\cB}\Gamma_D e^{\cB}\right)=\frac{N^2}{2\vert \Lambda\vert }-\frac{1}{2\vert \Lambda\vert }\Tr( e^{-\cB}\cN_+e^{\cB}\Gamma_D )^2 \leq \frac{N^2}{2|\Lambda|}.
\end{equation}
We now apply the Bogoliubov transformation to \cref{eq:H_weyl} using \cref{eq. bogolibov application} and commuting to normal order. We also use once more that $\Gamma_D$ is the Gibbs state of the second quantization of a one-body operator so that $\tr(a_p^*a_q^*\Gamma_D) =0$ for all $p,q$ and $\tr(a_p^*a_q\Gamma_D) = 0$ for $p\neq q.$ A straightforward calculation then shows the upper bound
\begin{align}\label{eq: first bound in diagonilization}
    &\Tr(\cH_v^{\rm int} \Gamma_B) \leq \hat{v}(0)\frac{ \rho N}{2} + \sum_{p \in \Lambda_+^*} \rho_0\hat{v}(p) \tr(e^{-\cB} a_p^* a_p e^{\cB} \Gamma_D) \nn
\\
&   \quad+\frac{N_0}{2|\Lambda|} \sum_{p \in \Lambda^*} \hat{v}(p) \tr(e^{-\cB}a_p^*a_{-p}^* e^{\cB}\Gamma_D) + \hc + \frac{1}{2|\Lambda|}\sum_{p,q,k \in \Lambda^*} \hat{v}(k)\tr(e^{-\cB}a_p^*a_{q+k}^*a_qa_{p+k}e^{\cB}\Gamma_D) \nn
\\
&\leq \widehat{v}(0)\frac{\rho N}{2}+\frac{1}{2\vert\Lambda\vert}\sum_{p,k\in\Lambda_+^*}\widehat{v}(k)c_ps_pc_{p+k}s_{p+k}+\sum_{p\in\Lambda_+^*} \left[\rho_0\widehat{v}(p) s_p^2+\rho_0\widehat{v}(p)s_pc_p\right] \nn
    \\ &\quad +\sum_{p\in\Lambda_+^*}\left(  \rho_0\widehat{v}(p)(s_p^2+c_p^2)+2\rho_0\widehat{v}(p)c_ps_p+\frac{2}{\vert \Lambda\vert}\sum_{\substack{k\in\Lambda^*: \\ p+k\neq 0}}\widehat{v}(k)c_ps_pc_{p+k}s_{p+k}\right)\Tr(a_p^*a_p\Gamma_D) \nn
    \\
    & \quad+C\frac{\hat{v}(0)}{|\Lambda|}\left(\sum_{p\in\Lambda_+^*} s_p^2 + \sum_{p\in\Lambda_+^*} (s_p^2 + c_p^2)\tr(a_p^*a_p \Gamma_D) \right)^2,
\end{align}
where we also used the simple inequalities $\vert \widehat{v}(k)\vert \leq \widehat{v}(0)$ and $s_q^2\leq c_q^2.$
From the inequalities in \Cref{Lem. Usefull bounds} we find \Cref{Lem:Hamil_energy_bound_general}.
\end{proof}

\subsection{Soften the Hamiltonian}\label{sect:softening}
In this subsection we employ the Jastrow factor \cref{Jastrow factor}, which softens the potential.
Due to the positive temperature we need to keep track on the higher order excitations the Jastrow factor creates.
Throughout this subsection, we explicitly display the potential in the Hamiltonian as $\cH_v$, recall \eqref{eq:def_cH}. 

\begin{lemma}\label{from hard to soft} 
Let $v$ be as in \Cref{thm:main} and $\Gamma$ a periodic density matrix on $\cF$.
Let $F$ be the Jastrow factor from \cref{Jastrow factor}. Then
	\begin{equation*}
	\tr(\cH_v F\Gamma F)\leq \tr(\Gamma\cH_{\tilde v}) - \tr(\Gamma \mathcal{R}),
	\end{equation*}
where
\[ 
\tilde{v}=(-2\Delta+v)f=\frac{2}{b\log(\frac{b}{\fa})}\delta_{\vert x\vert =b} \quad \text{ and }
\quad\mathcal{R} = \bigoplus_{n=0}^\infty R_n, \, R_n = F_n^2\sum_{i \neq j \neq k}^n \frac{(\nabla f)_{ij}}{f_{ij}}\cdot \frac{(\nabla f)_{ik}}{f_{ik}}
\]
with the notation
$ \{i \neq j \neq k\} = \{$pairwise distinct indices  $i, j, k$ running from $1$ to $n\}$.
\end{lemma}
\begin{proof}
Let $P_n$ be the projection onto the $n-$particle sector $L_s^2(\Lambda^n)$. Then $P_n$ commutes with $\cH$ and $F$ so that
\begin{equation} \label{eq:soften_fock}
	\tr(\cH_v F\Gamma F)= \sum_{n=0}^\infty \tr(F_n \cH_v F_n P_n\Gamma P_n).
\end{equation}
$P_n\Gamma P_n$ is a density matrix on $L_s^2(\Lambda^n)$ thus we may write $P_n\Gamma P_n=\sum_{m\in \mathbb{N}} \lambda_{n,m}\vert \phi^n_m\rangle \langle \phi_m^n\vert $ for some $\phi_m^n \in L_s^2(\Lambda^n), \lambda_{n,m}\geq 0$.
Therefore, by linearity, it is enough to consider the following expression for $\phi \in L_s^2(\Lambda^n)$
\begin{align} \label{eq:soften_nsector}
&\tr (F_n \cH_v F_n\vert \phi \rangle \langle \phi\vert)=\braket{F_n\phi, \cH_v F_n\phi}
\\
	&=\sum_{i=1}^n\int_{\Lambda^n} F_n^2\vert \nabla_i \phi\vert^2+ \vert {\phi}\vert^2\vert \nabla_i F_n\vert^2+ (\overline{\phi}\nabla_i\phi+\phi\nabla_i\overline{\phi})\cdot F_n\nabla_iF_n 
	+\sum_{i<j}^n\int_{\Lambda^n}v(x_i-x_j)F_n^2\vert \phi\vert^2. \nn
\end{align}
Partial integration yields
\begin{equation} \label{eq:soften_PI}
	\int_{\Lambda^n}\overline{\phi}\nabla_i\phi\cdot F_n\nabla_iF_n=-\int_{\Lambda^n}\phi\nabla_i\overline{\phi}\cdot F_n\nabla_iF_n-\vert {\phi}\vert^2\vert \nabla_i F_n\vert^2+\vert \phi\vert^2F_n(-\Delta_i F_n).
\end{equation}
From the definition $F_n = \prod_{i<j}^n f_{ij}$ we find
\begin{equation} \label{eq:soften_kinetic}
-\sum_{i=1}^n \Delta_i F_n=F_n\sum_{i\neq j}^n \frac{-\Delta_if_{ij}}{f_{ij}} - F_n\sum_{i\neq j\neq k}^n \frac{\nabla_i f_{ij}}{f_{ij}}\cdot \frac{\nabla_if_{ik}}{f_{ik}}.
\end{equation}
We insert  \eqref{eq:soften_PI} and  \eqref{eq:soften_kinetic} into \eqref{eq:soften_nsector}, use $F_n \leq f_{ij} \leq 1$ and obtain
\begin{align*}
\tr(F_n\cH_vF_n |\phi\rangle\langle\phi|) &\leq \sum_{i=1}^n \braket{\phi, -\Delta_i \phi} + \sum_{i<j}^n \int_{\Lambda^n} |\phi|^2 \big((v-2\Delta)f\big)(x_i-x_j) - \braket{\phi, R_n \phi}
\\
&= \tr(\cH_{\tilde v}|\phi\rangle\langle \phi|) - \tr(R_n |\phi\rangle\langle \phi|) .
\end{align*}
Inserting this into \eqref{eq:soften_fock} concludes the proof.
\end{proof}
The next lemma shows that we can bound the error $\mathcal{R}$ and the difference between $\Tr(F\Gamma_BF)$ and $\Tr(\Gamma_B)$.

\begin{lemma}\label{Lem: bound on errors involving F}
Let $v$ be as in \Cref{thm:main}, $\Gamma_B$ the trial state \cref{eq:def_Gamma_B}, $F$ the Jastrow factor from \cref{Jastrow factor}, $\mathcal{R}$ the operator defined in \Cref{from hard to soft} and $T$ as in \Cref{Lem. Usefull bounds}. Then
\begin{align}
    \Tr(F\Gamma_BF) &\geq 1-C b^2\rho N, \label{eq: denominator bound}\\
     \Tr(\mathcal{N}F\Gamma_BF)&\geq N(1-C\rho b^2 N),\label{eq. Bound on number of particles after Jastrow} \\
    \pm Tr(\mathcal{R}\Gamma_B)&\leq Cb^2\rho^2 N. \label{eq: bound on R}
\end{align}
\end{lemma}

\begin{proof}
    From the inequality
    \begin{equation}\label{eq. Dyson inequality}
    F_n^2=\prod_{i<j}^n(1-(1-(f_{ij})^2))\geq 1-\sum_{i<j}^n1-(f_{ij})^2=:1-\sum_{i<j}^nu_{ij},
    \end{equation}
we find
\begin{equation*}
    \Tr(F\Gamma_BF)\geq 1-\Tr(\Gamma_B \cH^{\rm int}_u),
\end{equation*}
where the operator $\cH^{\rm int}_u$ was defined in \eqref{eq:def_Hint}.
From $f\geq \1_{\{|x|>b\}}$ we obtain $\hat{u}(0) \leq \pi b^2$ so that \Cref{Lem:Hamil_energy_bound_general} implies \eqref{eq: denominator bound}.

Let us turn our attention to \cref{eq. Bound on number of particles after Jastrow}. From the fact that $F$ commutes with the number operator and inequality \cref{eq. Dyson inequality} we obtain
\begin{equation*}
\Tr(\mathcal{N}F\Gamma_BF)\geq \Tr(\mathcal{N}\Gamma_B)-\Tr(\mathcal{N}\cH^{\rm int}_u\Gamma_B) = N - \Tr(\mathcal{N}\cH^{\rm int}_u\Gamma_B).
\end{equation*}
By a straightforward computation similar to \eqref{eq: first bound in diagonilization} we find
\begin{equation*}
    \Tr(\mathcal{N}U\Gamma_B)\leq C \frac{\widehat{u}(0)}{\vert \Lambda\vert}\left(N+\sum_{p\in\Lambda_+^*} (c_p^2+s_p^2)\Tr(a_p^*a_p\Gamma_D)\right)^3 \leq C b^2 \rho N^2,
\end{equation*}
and \eqref{eq. Bound on number of particles after Jastrow} follows.

To show \cref{eq: bound on R} we use $|R_n| \leq \sum_{i\neq j\neq k}^n |(\nabla f)_{ij}||(\nabla f)_{ik}|$, which implies the following inequality in the second quantization formalism
\begin{equation*}
    \pm \mathcal{R}
    \leq \frac{1}{\vert \Lambda\vert^2 }\sum_{p,q,r,s,k \in \Lambda^*}\hat{|\nabla f|}(k)\widehat{|\nabla f|}(s)a_p^*a_q^*a_r^*a_{r-s}a_{q-k}a_{p+k+s}.
\end{equation*}
Again by a straightforward computation similar to \eqref{eq: first bound in diagonilization} we find 
\begin{equation*}
        \pm \Tr(\Gamma_B\mathcal{R})\leq C\frac{\widehat{|\nabla f|}(0)^2}{\vert \Lambda\vert^2}\left(N+\sum_p (c_p^2+s_p^2)\Tr(a_p^*a_p\Gamma_D)\right)^3.
\end{equation*}
From a Cauchy--Schwarz inequality and the definition of $f$ we find
\begin{equation*}
    \widehat{|\nabla f|}(0)^2 \leq C b^2\int_{|x|\leq b} \vert \nabla f\vert^2 \leq Cb^2,
\end{equation*}
where the last inequality used $f=\phi_b$ on $\{|x|\leq b\}$ and that $\phi_b$ is the minimizer of \eqref{def:scattering length}. \eqref{bound on Q2} then implies \eqref{eq: bound on R}.
\end{proof}

\subsection{Diagonalizing the soft Hamiltonian}

\begin{lemma} \label{Lemma: diagonalization of hamiltonian}
Let $\cH_{\tilde{v}}$ be the Hamiltonian with 
$$
\tilde{v}=(-2\Delta+v)f=\frac{2}{b\log(\frac{b}{\fa})}\delta_{\vert x\vert =b},
$$
and $b\leq \rho^{-\frac{1}{2}}$. Let $\Gamma_B$ be the trial state given by \cref{eq:def_Gamma_B} and $T$ as in \Cref{Lem. Usefull bounds}. Let $\alpha > 3/2.$
Then there exist a constant $C>0$ such that
\begin{align*}
        \Tr(\mathcal{H}_{\tilde{v}}\Gamma_B) &\leq 2\pi N \rho \delta\left(1+\delta\Big(\gamma + \frac{1}{4}+\frac{\log(\pi)}{2}\Big) \right)+\Tr(H_D\Gamma_D) 
        \\
        &\quad + CN\rho \frac{1}{\log\left(\frac{b}{\fa}\right)}\left(1-\delta \log\Big(\frac{b}{\fa}\Big)+Y+\frac{T}{T_c}\right)^2.
    \end{align*}
\end{lemma}

\begin{proof}
For notational convenience, we drop the tilde and simply write \( v = \tilde{v} \) as well as $\mathcal{H} = \cH_{\tilde{v}}$.
Recall the computation \eqref{eq: first bound in diagonilization}. Together with 
$$
\weyl^* \sum_{p\in\Lambda^*} p^2a_p^*a_p \weyl = \sum_{p\in\Lambda^*_+} p^2a_p^*a_p
$$
we find
\begin{align}\label{eq. first bound in diagonilization}
    &\Tr(\cH \Gamma_B) \leq \widehat{v}(0)\frac{N\rho}{2}+\frac{1}{2\vert\Lambda\vert}\sum_{p,k\in\Lambda_+^*}\widehat{v}(k)c_ps_pc_{p+k}s_{p+k}+\sum_{p\in\Lambda_+^*} \left[p^2s_p^2+\rho_0\widehat{v}(p) s_p^2+\rho_0\widehat{v}(p)s_pc_p\right] \nn
    \\ & +\sum_{p\in\Lambda_+^*}\left(  (p^2+\rho_0\widehat{v}(p))(s_p^2+c_p^2)+2\rho_0\widehat{v}(p)c_ps_p+\frac{2}{\vert \Lambda\vert}\sum_{\substack{k\in\Lambda^*: \\ p+k\neq 0}}\widehat{v}(k)c_ps_pc_{p+k}s_{p+k}\right)\Tr(a_p^*a_p\Gamma_D) \nn
    \\
    & +C\frac{\hat{v}(0)}{|\Lambda|}\left(\sum_{p\in\Lambda_+^*} s_p^2 + \sum_{p\in\Lambda_+^*} (s_p^2+c_p^2)\tr(a_p^*a_p \Gamma_D) \right)^2,
\end{align}
Let us simplify the constant term in \eqref{eq. first bound in diagonilization}. First, we note that we can replace all sums in the first line by integrals up to an error of order 
$$\hat{v}(0) L^{-1}N\rho^{1/2}Y^{1/2}|\log(Y)| =\hat{v}(0) \rho N Y^{\alpha+1/2}|\log(Y)|.$$
This approximation follows from the standard technique of passing from sums to integrals in the limit; we use the mean value theorem and can bound the error in terms of bounds on the gradient of the considered function. For further details, see \cite[Appendix G]{FGJMO-22}. In particular, this estimate makes use of \eqref{eq:sc_bound}.
Secondly, as mentioned earlier, we expect $c_ps_p \approx -\rho_0\hat{\omega}(p)$ and therefore we write
\begin{align*}
&\frac{1}{2}\int_{\R^2\times\R^2} \hat{v}(k) c_ps_p  c_{p+k}s_{p+k} \dd p \dd k = \frac{1}{2}\braket{cs,\hat{v}\ast (cs)} \\\
&= \frac{1}{2}\braket{cs+\rho_0\hat{\omega},\hat{v}\ast(cs+\rho_0\hat{\omega})} -\braket{cs,\hat v\ast\rho_0\hat{\omega}}-\frac{1}{2}\braket{\rho_0\hat{\omega},\hat{v}\ast\rho_0\hat{\omega}}
\\
&= \frac{1}{2}\braket{cs+\rho_0\hat{\omega},\hat{v}\ast(cs+\rho_0\hat{\omega})} -\rho_0 (2\pi)^2 \braket{cs,\hat {v\omega}}-\frac{\rho_0^2(2\pi)^4}{2} \hat{v\omega^2}(0).
\end{align*}
According to \cite[Lemma 4.8]{FGJMO-22} the first term in the last line is bounded by $C\hat{v}(0)\rho^2 Y^2$ so that
\begin{align*}
        &\widehat{v}(0)\frac{N\rho}{2}+\frac{1}{2\vert\Lambda\vert}\sum_{p,k\in\Lambda_+^*}\widehat{v}(k)c_ps_pc_{p+k}s_{p+k}+\sum_{p\in\Lambda_+^*} \left[ p^2s_p^2+\rho_0\hat{v} (p) s_p^2+\rho_0\widehat{v}(p)s_pc_p \right]
        \\
        &\leq \widehat{v}(0)\frac{N\rho}{2} - \frac{N_0\rho_0}{2}\hat{v\omega^2}(0) + |\Lambda|\int_{\R^2} \left[ p^2s_p^2 + \rho_0\hat{v}(p)s_p^2 + \rho_0 c_p s_p \hat{g}(p)\right] \frac{\dd p}{(2\pi)^2} \nn 
        \\
        &\qquad + C \hat{v}(0) \rho N Y^2(Y^{\alpha-3/2}|\log(Y)|+ 1) \nn 
        \\
        &= \hat{g}(0) \frac{N\rho}{2} +\frac{N_0\rho_0}{2} \hat{g\omega}(0) + |\Lambda|\int_{\R^2} \left[ p^2s_p^2 + \rho_0\hat{g}(p)s_p^2 + \rho_0c_ps_p\hat{g}(p) \right] \frac{\dd p}{(2\pi)^2} \\
        & \qquad + C\hat{v}(0) \rho N Y^2(Y^{\alpha-3/2}|\log(Y)|+ 1) + \frac{N\rho-N_0\rho_0}{2}\hat{v\omega}(0) + |\Lambda|\rho_0 \int_{\R^2}(\hat{v}(p)-\hat{g}(p)) s_p^2 \frac{\dd p}{(2\pi)^2} 
        \\
        &\leq 2\pi N\rho\delta + \frac{|\Lambda|}{2}\left(\rho_0^2\hat{g\omega}(0) + \int_{\R^2} D_p-p^2-\rho_0\hat{g}(p) \frac{\dd p}{(2\pi)^2}\right)
        \\
        &\qquad + C \hat{v}(0) \rho N Y^2(Y^{\alpha-3/2}|\log(Y)|+ 1) + C \rho (N-N_0 + NY) (\hat{v}(0)-\hat{g}(0)),
\end{align*}
where we used \eqref{eq:g_hat_0},
$$
p^2s_p^2+\rho_0\widehat{g}(p)s_p^2+\rho_0\widehat{g}(p)s_pc_p=\frac{1}{2}(D_p-p^2-\rho_0\widehat{g}(p))
$$
and \eqref{bound on s^2}.
We use the integral representation of $\hat{g\omega}(0)$, see \eqref{eq:gw_hat}, and we may replace $\hat{g}(p)$ by $\hat{g}(0)$ everywhere up to an error of order
$
\rho^2 Y^3
$, see \cite[Proposition C.3]{FGJMO-22}. Then we can explicitly compute
\begin{align*}
&\frac{1}{2}\int_{\R^2} D_p-p^2-\rho_0\hat{g}(0)+\rho_0^2\frac{\hat{g}(0)^2}{2p^2}\1_{\{|p|>2e^{-\gamma-\delta^{-1}}\fa^{-1}\}} \frac{\dd p}{(2\pi)^2} 
\\
&= \frac{(\rho_0\hat{g}(0))^2}{16\pi}\left(2\gamma + \frac{1}{2}+\log(\pi)+\log(\frac{\rho_0 \delta}{2\rho Y})\right)
\\
&\leq 2\pi \rho^2 \delta^2 \left(\gamma + \frac{1}{4}+\frac{\log(\pi)}{2}\right),
\end{align*}
where we used $\rho_0\leq \rho, \delta\leq 2Y.$
From $b\leq \rho^{-\frac{1}{2}}$, $\alpha> 3/2$, \eqref{eq: difference between v and g} and \Cref{Lem. Usefull bounds} we obtain
\begin{equation} \label{eq:const_term_final}
\begin{split}
    &\hat{v}(0)\frac{N\rho}{2}+\frac{1}{2\vert\Lambda\vert}\sum_{p,k\in\Lambda_+^*}\widehat{v}(k)c_ps_pc_{p+k}s_{p+k}+\sum_{p\in\Lambda_+^*} p^2s_p^2+\rho_0\widehat{v}(p) s_p^2+\rho_0\widehat{v}(p)s_pc_p
    \\
    &\leq 2\pi N\rho \delta \left(1 + \delta \Big(\gamma + \frac{1}{4}+\frac{\log(\pi)}{2}\Big)\right)+C\rho N\frac{1}{\log\left(\frac{b}{\fa}\right)} \bigg(Y + \frac{T}{T_c}\bigg) \bigg( 1-\delta \log\Big(\frac{b}{\fa}\Big) +Y\bigg).
\end{split}
\end{equation}

Returning to the second line of \cref{eq. first bound in diagonilization} we find that
\begin{equation} \label{eq:vcs_replace}
    \Big|\frac{1}{\vert \Lambda\vert}\sum_{\substack{k\in\Lambda^*:\\ k+p\neq 0}}[\widehat{v}(k)c_{p+k}s_{p+k}] + \rho_0\widehat{v\omega}(p)\Big| 
    \leq C\widehat{v}(0)\rho Y + C\hat{v}(0) \rho Y^{\alpha + 1/2} |\log (Y)|,
\end{equation}
which follows as before by replacing sums with integrals and using that
$$
|\braket{cs + \rho_0\hat{\omega},\hat{v}(p-\cdot)}| \leq C\hat{v}(0)\rho Y
$$
uniformly in $p$, which is shown in the first step of the proof of \cite[Lemma 4.7]{FGJMO-22}.
Inserting \eqref{eq:vcs_replace} in the second line of \cref{eq. first bound in diagonilization} and using 
\begin{equation}
    (p^2+\rho_0\widehat{g}(p))(s_p^2+c_p^2)+2\rho_0\widehat{g}(p)c_ps_p=D_p
\end{equation}
yields 
\begin{align*}
 &\sum_{p\in\Lambda_+^*}\left(  (p^2+\rho_0\widehat{v}(p))(s_p^2+c_p^2)+2\rho_0\widehat{v}(p)c_ps_p+\frac{2}{\vert \Lambda\vert}\sum_{\substack{k\in\Lambda^*: \\ p+k\neq 0}}\widehat{v}(k)c_ps_pc_{p+k}s_{p+k}\right)\Tr(a_p^*a_p\Gamma_D)
\\
&\leq \sum_{p\in\Lambda_+^*} D_p \Tr(a_p^*a_p\Gamma_D) 
\\
&\qquad + C \rho \left((\hat{v}(0) -\hat{g}(0)) + \widehat{v}(0)Y + \hat{v}(0) Y^{\alpha + 1/2}|\log (Y)|\right) \sum_{p\in\Lambda_+^*} (s_p^2+c_p^2) \Tr(a_p^*a_p\Gamma_D)
\\
&\leq \Tr(H_D\Gamma_D) 
 + C \rho N \frac{1}{\log\left(\frac{b}{\fa}\right)}  \frac{T}{T_c}\left(1-\delta \log\Big(\frac{b}{\fa}\Big)+Y \right),
\end{align*}
where we used $\alpha> 3/2$, \Cref{Lem. Usefull bounds}, \eqref{eq: difference between v and g} and $D_p\leq \sqrt{p^4+2\rho \widehat{g}(0)p^2}$ for the last inequality.
Inserting this and \eqref{eq:const_term_final} in \eqref{eq. first bound in diagonilization} implies the result.
\end{proof}

\begin{proof}[Proof of \Cref{lem:trial_hamil}]
Recall the definition of $\Gamma_0$ from \cref{eq:def_Gamma_0}. 
We choose $\eta$ such that \cref{eq:N_Gamma} holds, i.e., such that
\begin{equation} \label{eq:eta_choice}    
|\Lambda|^{-1} \tr(\cN {\Gamma}_0) = \tilde\rho.
\end{equation}
Let us show that this is possible. By Stones theorem the map
$ \eta \mapsto \weyleta \Psi$ is continuous. 
Together with the fact that $F$ commutes with $\mathcal{N}$, a computation shows that  $\eta \mapsto \Tr(\mathcal{N}\Gamma_0)$ is continuous.
With \eqref{eq: denominator bound} we obtain that 
$$
|\Lambda|^{-1} \tr(\cN {\Gamma}_0) \leq \frac{N + N_0\eta}{|\Lambda|}(1+Cb^2\rho N) 
\leq \rho,
$$
given that $\eta \leq -Cb^2\rho N$. Here we in particular used $\rho_0 \geq \rho/2$ and the smallness condition on $b$ \cref{equation smallness of b}.
On the other hand from \eqref{eq. Bound on number of particles after Jastrow} we obtain
$$
|\Lambda|^{-1} \tr(\cN {\Gamma}_0) \geq \frac{N + N_0\eta -C\rho b^2N^2}{|\Lambda|} 
\geq \rho(1+Y^2),
$$
given that $\eta \geq C(b^2\rho N+Y^2)$. Therefore, there exits an $\eta$ with $|\eta|\leq C(b^2\rho N+Y^2)$ such that \eqref{eq:eta_choice} holds.

Then by Lemmas~\ref{from hard to soft} and \ref{Lem: bound on errors involving F} as well as \eqref{eq:H_weyl} and \Cref{Lem. Usefull bounds} we obtain
\begin{align} \label{eq:est_Lem2_1}
    &\tr(\cH_v {\Gamma}_0) \leq \frac{1}{1-Cb^2\rho N}\left(\tr(\cH_{\tilde{v}} \weyleta e^\cB \Gamma_D e^{-\cB} \weyleta^*) + C b^2 \rho^2 N\right) \nn
    \\
    &= \frac{1}{1-Cb^2\rho N}\bigg(\tr(\cH_{\tilde v} \Gamma_B) + \frac{N_0^2((1+\eta)^2-1)}{2|\Lambda|}\hat{\tilde v}(0)  + \frac{N_0\eta}{|\Lambda|} \sum_{p\in\Lambda^*}( \hat{\tilde v}(0) + \hat{\tilde v}(p)) \tr( a_p^*a_p e^{\cB}\Gamma_D e^{-\cB}) \nn
    \\
    &\qquad + \frac{N_0 \eta}{2|\Lambda|} \sum_{p \in \Lambda^*} \hat{\tilde{v}}(p) \tr(a_p^*a_{-p}^* e^{\cB}\Gamma_D e^{-\cB}) + \hc  + C b^2 \rho^2  N \bigg)
    \\
    &\leq \tr(\cH_{\tilde v} \Gamma_B)(1+Cb^2\rho N) + C \frac{1}{\log\left(\frac{b}{\fa} \right)} |\eta| N \rho   + C b^2 \rho^2 N. \nn
\end{align}
\Cref{Lemma: diagonalization of hamiltonian} and the straightforward bound 
$\tr(H_D\Gamma_D)
\leq CN\rho $ 
yields the result.

\end{proof}

\section{Proof of Lemma~\ref{lem:trial_entropy_easy}}
In the previous section we saw that we could soften the potential by introducing the Jastrow factor $F$. In this section, we show that the Jastrow factor does not change the entropy too much. In the literature, similar results are typically established as in \cite[Lemma 2]{seiringer_thermodynamic_2006}.
Our approach differs; we use an argument based on a Fannes-type inequality \cite{Fannes_1973}; see also Theorem 11.6 in \cite{nielsen_quantum_2012}. This allows us to bound the difference in entropy in terms of the trace norm distance. An advantage of this method is that it permits working directly with the full trial state and normalizing it in a natural way, rather than normalizing each eigenfunction individually.
Since the Fannes inequality applies only in finite-dimensional settings, we introduce a truncation of the density matrices. More precisely, we restrict to eigenfunctions with less than $M$ particles in each state and no particles with momenta higher than $\sqrt{K}$. This leads to the following lemma, from which \Cref{lem:trial_entropy_easy} follows by a suitable choice of $K$ and $M$.
\begin{lemma}\label{lem:trial_entropy}
Let $K\geq 0$, $M \in \mathbb{N}$ and $T$ as in \Cref{Lem. Usefull bounds}.
Let 
$$
\Gamma_W := \weyleta e^\cB \Gamma_D e^{-\cB} \weyleta^*
$$
so that 
$$
\Gamma_0 = \frac{F\Gamma_WF}{\tr(F\Gamma_WF)}
$$ 
is our trial state from \cref{eq:def_Gamma_0}.
There exists a constant $C>0$ such that
\begin{equation} \label{eq:trace_difference}
    \|\Gamma_0 - \Gamma_W \|_1  \leq C\sqrt{b^2 \rho N} + Cb^2\rho N.
\end{equation}
Moreover, if $\| \Gamma_0 - \Gamma_W\|_1 \leq \frac{1}{e}$, then
\begin{align} \label{eq:entropy_diff_full}
    -TS(\Gamma_0) \leq &-TS(\Gamma_W) + C L^2T^2(L^2T+1)\left(M^{-1}+e^{-\beta K/2}\right)+ CL^2T^2 \beta K \log(M) \|\Gamma_0 - \Gamma_W\|_1 \nn
    \\
    &   + T h(\|\Gamma_0 - \Gamma_W\|_1),
\end{align}
with $h(x) = -x\log(x),\, x\in [0,1]$.

\end{lemma}
\begin{proof}[Proof of \Cref{lem:trial_entropy_easy}]
Let $M = \lceil Y^{-1/2 }L/b \rceil, \beta K = 2\log(M)$. Together with $S(\Gamma_W) = S(\Gamma_D)$ we obtain \Cref{lem:trial_entropy_easy}.
\end{proof}

\begin{proof} [Proof of \Cref{lem:trial_entropy}]
Let us start by showing \eqref{eq:trace_difference}, i.e., that $\Gamma_0$ and $\Gamma_W$ are close in trace norm.
Due to the spectral theorem we may write 
$$
\Gamma_W = \sum_{i \geq 1} \lambda_i |\Psi_i\rangle\langle \Psi_i|
,\quad \Gamma_0 = \frac{1}{\tr(F\Gamma_W F)} \sum_{i \geq 1} \lambda_i |F \Psi_i\rangle\langle F\Psi_i| = 
\sum_{i \geq 1} \lambda_i^F |\Psi_i^F\rangle\langle \Psi_i^F|
$$
for decreasing sequences of eigenvalues $\lambda_i, \lambda_i^F \geq 0, \, i \in \mathbb{N}$ and normalized eigenfunctions $\Psi_i, \Psi_i^F \in \cF(L^2(\Lambda))$ of $\Gamma_W$ respectively $\Gamma_0$.
The triangle inequality yields
\begin{align} \label{eq:tr_difference_triangle}
\|\Gamma_0 - \Gamma_W\|_1 &\leq  |1-\tr(F\Gamma_W F)| \, \|\Gamma_0\|_1
+ \sum_{i\geq 1} \lambda_i \, \| |\Psi_i\rangle\langle \Psi_i| - |F \Psi_i\rangle\langle F\Psi_i| \|_1 .
\end{align}
We further estimate the second term on the right-hand side as follows
\begin{align*}
\||\Psi_i\rangle\langle \Psi_i| - |F \Psi_i\rangle\langle F\Psi_i| \|_1 
\leq 2\| (1-F) \Psi_i\|
\leq 2\sqrt{\braket{\Psi_i, \cH_{1-f}^{\rm int} \Psi_i}},  
\end{align*}
where we used $1-F_n \leq \sum_{i<j}^n 1-f_{ij}$ for $n\in\bN$, see \eqref{eq. Dyson inequality}. Then, applying the Jensen inequality in \eqref{eq:tr_difference_triangle} gives
$$
\|\Gamma_0 - \Gamma_W\|_1 \leq |1-\tr(F\Gamma_W F)| +2\sqrt{\tr\left(\cH_{1-f}^{\rm int}  \Gamma_W\right)}.
$$
We bound the first term using \Cref{eq: denominator bound}. The second term can be controlled as in \Cref{Lem:Hamil_energy_bound_general}. The inequality $f \geq \1_{\{|x|>b\}}$ yields
$
\hat{1-f}(0) 
\leq \pi b^2 .
$
We obtain \cref{eq:trace_difference}.

Next, we want to show the entropy inequality \eqref{eq:entropy_diff_full}. As explained before we  need to truncate the problem to a finite dimensional one. Thus let $I \subset \mathbb{N}$ a finite subset and note that the entropy function $h(x)=-x\log(x)$ is concave and positive on the unit interval $[0,1]$. We find 
\begin{align} \label{eq:entropy_inequality}
S(\Gamma_0) = \sum_{i \geq 1} h(\lambda_i^F) \geq \sum_{i\in I} h(\lambda_i^F) = S(\Gamma_W) - \sum_{i\notin I} h(\lambda_i) + \sum_{i\in I}\big[ h(\lambda_i^F)-h(\lambda_i)\big].
\end{align}
Our choice of $I$ will be such that the second term on the right-hand side is small, that is, $I$ includes all entropy-relevant momenta. The third term is bounded with a Fannes type inequality and depends only logarithmically on the cardinality of $I$.

Let us start with the third term. We want to show the following Fannes type inequality.
\begin{align} \label{eq:entropy_diff_I}
\left|\sum_{i\in I}\big[h(\lambda_i^F)-h(\lambda_i)\big]\right| 
\leq \log(|I|) \|\Gamma_0 - \Gamma_W\|_1  + h\big(\|\Gamma_0 - \Gamma_W\|_1\big).
\end{align}
Since we are not in the setting of the reference \cite{Fannes_1973} we provide a complete proof.
As a first step we show 
\begin{equation}\label{eq:ev_difference}
\sum_{i\geq 1} |\lambda_i^F - \lambda_i| \leq \|\Gamma_0 - \Gamma_W\|_1.
\end{equation}
Indeed, this is a very general statement. We may write 
$\Gamma_0 - \Gamma_W = P - Q$, where $P,Q$ denote the projections of $\Gamma_0-\Gamma_W$ onto its positive and negative eigenspaces respectively.
Let $\mu_i(\cdot)$ denote the $i-$th largest eigenvalue of an operator. Then from the definition of $\lambda_i,\lambda_i^F$ we find $\lambda_i^F = \mu_i(\Gamma_0) \leq \mu_i(Q+\Gamma_0) = \mu_i(P+\Gamma_W) \geq \mu_i(\Gamma_W) = \lambda_i$. We obtain
$$
 \mu_i(Q+\Gamma_0) + \mu_i(P+\Gamma_W) \geq 2\max\{\lambda_i^F,\lambda_i\}=\lambda_i^F+\lambda_i+|\lambda_i^F - \lambda_i|.
$$
We conclude the proof of \eqref{eq:ev_difference} as follows
$$
\|\Gamma_0 - \Gamma_W\|_1 = \tr(P+Q) \geq \sum_{i \geq 1} \big[\lambda_i^F+\lambda_i+|\lambda_i^F-\lambda_i|\big] - \tr(\Gamma_0 + \Gamma_W) = \sum_{i\geq 1}|\lambda_i^F-\lambda_i|.
$$
In particular, $|\lambda_j^F-\lambda_j| \leq \sum_{i\in I} |\lambda_i^F - \lambda_i| \leq  \|\Gamma_0 - \Gamma_W\|_1 \leq 1/e$ for all $j \in I$, where we used the assumption from the lemma.
Using this, together with the elementary observations that $h$ is concave and monotone on $[0,1/e]$, and that for $x,y\in [0,1]$ with $|x-y|\leq \frac{1}{2}$ we have $|h(x)-h(y)| \leq h(|x-y|),$ we can now show \eqref{eq:entropy_diff_I}.
\begin{align*}
\left|\sum_{i\in I}\big[h(\lambda_i^F)-h(\lambda_i)\big]\right| &\leq  |I| \, h \left(\frac{\sum_{i\in I}|\lambda_i^F-\lambda_i|}{|I|}\right) = \log(|I|) \sum_{i\in I}|\lambda_i^F-\lambda_i|  + h\left(\sum_{i\in I} |\lambda_i^F-\lambda_i|\right) \nn
\\
&\leq \log(|I|) \|\Gamma_0 - \Gamma_W\|_1  + h\big(\|\Gamma_0 - \Gamma_W\|_1\big).
\end{align*}

As a last step, we show that the second term on the right-hand side of \eqref{eq:entropy_inequality} is small, i.e., that we did not make too much of an error when truncating the density matrices. 
We use that we know the eigenvalues of $\Gamma_W$ explicitly; since $\Gamma_W$ is unitarily equivalent to 
$$
\Gamma_D = \frac{e^{-\beta H_D} \1_{\{\cN_0 = 0\}}}{\tr_{\cF_+}\left(e^{-\beta H_D}\right)},
$$
each eigenvalue $\lambda_i$ of $\Gamma_W$ uniquely corresponds to a finite sequence $\{n_p\}_{p\in \Lambda_+^*}, n_p \in \mathbb{N}_0$ with  
$$
\lambda_i = \frac{e^{-\beta \sum_{p\in \Lambda_+^*}D_p^0 n_p}}{\tr_{\cF_+}\left(e^{-\beta \sum_{p\in \Lambda_+^*}D_p^0 a_p^*a_p}\right)},
$$
where we adopted the short hand notation
\begin{equation*}
    D_p^0=\sqrt{p^4+2\widehat{g}(0)\rho  p^2}=\sqrt{p^4+8\pi \rho \delta p^2}.
\end{equation*}
For $K\geq 0, M\in \mathbb{N}$ we define $\{n_p\}_{p \in \Lambda_+^*} \in \mathcal{I}$ if and only if $n_p < M$ for all $p$ and $n_p = 0$ for all $|p|>\sqrt{K}$. We then define $i\in I $ if and only if its corresponding sequence $\{n_p\}_{p \in \Lambda_+^*} \in \mathcal{I}$.
In other words, $i\in I$ if and only if its corresponding eigenfunction has less than $M$ particles in each state and no particles with momenta higher than $\sqrt{K}$. 
We compute 
\begin{equation} \label{eq:sum_notinI}
\begin{aligned}
    \sum_{i \notin I} h(\lambda_i) 
&= \beta \frac{\sum_{\{n_p\}\notin \mathcal{I}} \sum_{p\in \Lambda_+^*} D_p^0n_p {e^{-\beta\sum_{q\in\Lambda_+^*} D_q^0n_q}}}{\tr_{\cF_+}\left(e^{-\beta \sum_{p\in \Lambda_+^*}D_p^0 a_p^*a_p}\right)}  
\\
&\qquad  + \log\left(\tr_{\cF_+}\left(e^{-\beta \sum_{p\in \Lambda_+^*}D_p^0 a_p^*a_p}\right)\right) \frac{ \sum_{\{n_p\}\notin \mathcal{I}} e^{-\beta \sum_{p\in \Lambda_+^*} D_p^0n_p}}{\tr_{\cF_+}\left(e^{-\beta \sum_{p\in \Lambda_+^*}D_p^0 a_p^*a_p}\right)}.
\end{aligned} \end{equation}
From 
$\mathcal{I}^c \subset \bigcup_{|q| \leq \sqrt{K}} \{n_q \geq M\} \cup \bigcup_{|q| > \sqrt{K}} \{n_q > 0\} $ it follows that
\begin{equation} \label{eq:sum_notinI_1} \begin{aligned}
\sum_{\{n_p\}\notin \mathcal{I}} e^{-\beta \sum_{p\in \Lambda_+^*} D_p^0n_p} 
&\leq \sum_{|q|\leq \sqrt{K}} \sum_{\substack{ \{n_p\}: \\ n_q \geq M}} e^{-\beta \sum_{p\in \Lambda_+^*} D_p^0n_p} + \sum_{|q| > \sqrt{K}} \sum_{\substack{\{n_p\}: \\ n_q > 0}} e^{-\beta \sum_{p\in \Lambda_+^*} D_p^0n_p} 
\\
&= \prod_{p \in \Lambda_+^*} \frac{1}{1-e^{-\beta D_p^0}} \left[ \sum_{|q| \leq \sqrt{K}} e^{-\beta D_q^0 M} + \sum_{|q| > \sqrt{K}}  e^{-\beta D_q^0}\right],
\end{aligned}\end{equation}
where the last step is a standard ideal gas computation. This ideal gas computation also shows that  $\prod_{p \in \Lambda_+^*} (1-e^{-\beta D_p^0})^{-1} = {\tr_{\cF_+}\Big(e^{-\beta \sum_{p\in \Lambda_+^*}D_p^0 a_p^*a_p}\Big)}$ so that in particular 
\begin{align} \label{eq:sum_notinI_2}
\log\left({\tr_{\cF_+}\left(e^{-\beta \sum_{p\in \Lambda_+^*}D_p^0 a_p^*a_p}\right)}\right) 
= -\sum_{p\in \Lambda_+^*}\log\left(1-e^{-\beta D_p^0}\right) 
\leq C \frac{L^2}{\beta},
\end{align}
where in the last step we bounded the sum by an integral. 
We may also bound the sums in \eqref{eq:sum_notinI_1} with integrals and find
\begin{align} \label{eq:sum_notinI_3}
\sum_{|q| \leq \sqrt{K}} e^{-\beta D_q^0 M} \leq C \frac{L^2}{\beta M} \;\; \text{ and } 
\sum_{|q| > \sqrt{K}}  e^{-\beta D_q^0} \leq C \frac{L^2}{\beta} e^{-\beta K}.
\end{align}
Combining \cref{eq:sum_notinI_1,eq:sum_notinI_2,eq:sum_notinI_3} shows that we may bound the second term on the right-hand side of \eqref{eq:sum_notinI} by $C L^2 T L^2 T \left(M^{-1} + e^{-\beta K}\right)$.
We proceed similar for the fist term on the right-hand side of \eqref{eq:sum_notinI} and find
\begin{align*}
&\beta \frac{\sum_{\{n_p\}\notin \mathcal{I}} \sum_{p\in \Lambda_+^*} D_p^0n_p {e^{-\beta\sum_{q\in\Lambda_+^*} D_q^0n_q}}}{\tr_{\mathcal{F}_+}\left(e^{-\beta \sum_{p\in \Lambda_+^*}D_p^0 a_p^*a_p}\right)}
\\ 
&\leq \beta\sum_{|q| \leq \sqrt{K}} \left(\sum_{p \in \Lambda_+^* \setminus\{q\}} \left[\frac{D_p^0}{e^{\beta D_p^0}-1}\right] e^{-\beta D_q^0 M} + D_q^0 \frac{\sum_{n=M}^{\infty} n e^{-\beta D_q^0 n}}{(1-e^{-\beta D_q^0})^{-1}} \right)
\\
& \qquad + \beta \sum_{|q| > \sqrt{K}} \left(\sum_{p \in \Lambda_+^* \setminus\{q\}} \left[ \frac{D_p^0}{e^{\beta D_p^0}-1}\right] e^{-\beta D_q^0} + D_q^0 \frac{\sum_{n=1}^{\infty} n e^{-\beta D_q^0 n}}{(1-e^{-\beta D_q^0})^{-1}} \right)  
\\
&\leq C L^2T \sum_{|q| \leq \sqrt{K}} e^{-\beta D_q^0 M} + \beta \sum_{|q| \leq \sqrt{K}} D_q^0 e^{-\beta D_q^0 M} \left(M + \frac{1}{e^{\beta D_q^0}-1}\right)
\\
&\qquad + C L^2T \sum_{|q| > \sqrt{K}} e^{-\beta D_q^0} + \beta \sum_{|q| > \sqrt{K}} \frac{D_q^0}{e^{\beta D_q^0}-1}
\\
&\leq CL^2T \left(\frac{L^2T+1}{M} + (L^2T+1)e^{-\beta K/2}\right).
\end{align*}
Inserting this into \eqref{eq:sum_notinI} concludes the bound of the second term in \eqref{eq:entropy_inequality} 
$$
\sum_{i \notin I} h(\lambda_i) \leq CL^2T (L^2T+1)(M^{-1}+e^{-\beta K/2}).
$$
Moreover, we can estimate $|I| \leq M^{C L^2K}$ so that from \eqref{eq:entropy_diff_I} we obtain 
$$
\left|\sum_{i\in I}\big[h(\lambda_i^F)-h(\lambda_i)\big]\right| \leq CL^2K \log(M) \|\Gamma_0 - \Gamma_W\|_1 + h(\|\Gamma_0 - \Gamma_W\|_1).
$$
Inserting this into \eqref{eq:entropy_inequality} yields \Cref{lem:trial_entropy}.
\end{proof}

\section{Proof of Proposition~\ref{prop:Localization}}
The proof of \Cref{prop:Localization} follows a standard argument and can, for the three-dimensional case, be found in \cite{haberbergerUpperBound2024}; see also \cite{basti2021newsecondorderupper}. The adaptation to two dimensions is straightforward. Here, we briefly highlight the steps that require modification.

The first lemma shows that a periodic state defined on a fixed box can be extended to a state satisfying Dirichlet boundary conditions on the thermodynamic box. The second lemma establishes that it is sufficient to consider grand canonical states to derive a corresponding statement for the canonical free energy.

\begin{lemma}\label{lem:Dirichlet}
Let $\Gamma_L$ be a normalized density matrix on $\cF(\Lambda_L)$ with periodic boundary conditions and $\tilde{L}=t(L+2\ell+R)$ with $t\in\mathbb{N}, \ell>0$ and $R>0$ such that $\supp(v)\subset B(0,R).$
Then there is a density matrix $\Gamma^D_{\tilde{L}}$ that satisfies Dirichlet boundary conditions on the thermodynamic box $\Lambda_{\tilde{L}}$ such that
\begin{align*}
&\tr \Gamma^D_{\tilde{L}} = 1,\quad \tr (\cN \Gamma^D_{\tilde{L}}) = t^2\tr(\cN \Gamma_L),\quad S(\Gamma^D_{\tilde{L}}) = t^2 S(\Gamma_L).
\intertext{and}
&\tr(\cH \Gamma^D_{\tilde{L}}) \leq t^2 \left( \tr(\cH \Gamma_L) + \frac{C}{L\ell}\tr(\cN \Gamma_L)\right). 
\end{align*}
\end{lemma}

\begin{proof}
\textit{From periodic to Dirichlet on fixed boxes; see \cite[Lemma A.1]{basti2021newsecondorderupper} and \cite[Lemma 27]{haberbergerUpperBound2024}.} 
We write $\Gamma_L$ in its eigenbasis 
$$
\Gamma_L = \sum_{j\in\bN} \lambda_j |\Psi_j\rangle\langle\Psi_j|.
$$
Then we extend each periodic eigenfunction $\Psi_j$ to a function with Dirichlet boundary conditions on the bigger box $\Lambda_{L+2\ell}$ by multiplying it in every spatial direction with a suitable cosine factor. 
Summing these modified eigenfunctions yields a state $\Gamma_{L+2\ell}^D$ on $\cF(\Lambda_{L+2\ell})$ with Dirichlet boundary conditions. 
Moreover, there exists a $\overline{u}\in \Lambda_L$ such that shifting the eigenfunctions yields a trial state $\Gamma^D_{L+2\ell,\overline{u}}$ on the box centered at $\overline{u}$.
One can verify that 
\begin{align*}
&S(\Gamma^D_{L+2\ell,\overline{u}}) = S(\Gamma_L),
\quad \tr (\cN^j\Gamma^D_{L+2\ell,\overline{u}} ) = \tr(\cN^j \Gamma_L ) \; \forall j\in\mathbb{N}.
\intertext{and}
&\tr(\cH \Gamma^D_{L+2\ell,\overline{u}}) \leq \tr(\cH \Gamma_L) + \frac{C}{L\ell}\tr(\cN \Gamma_L).
\end{align*}

\textit{Extension to the thermodynamic box. See \cite[Lemma A.2]{basti2021newsecondorderupper} and \cite[Chapter B.2]{haberbergerUpperBound2024}.} 
Starting from the previously constructed state, we replicate it $t^2$ times. Between adjacent boxes, we insert narrow corridors of width $R$ to ensure that there is no interaction between different copies. In this way, we obtain a thermodynamic trial state $\Gamma_{\tilde{L}}^D$.

Because the individual boxes do not interact, the computation of the particle number, the entropy, and the energy of $\Gamma_{\tilde{L}}^D$ are straightforward.

\end{proof}

\begin{lemma}{(From grand canonical to canonical)}\label{from grand canoncial to canonical}
Let $\Gamma_{\tilde L}^D $ be a sequence of normalized density matrices with Dirichlet boundary conditions on $\cF(\Lambda_{\tilde L})$ for $\tilde{L}\rightarrow \infty$ and
$$
 \tilde\rho =  \frac{\tr(\cN \Gamma^D_{\tilde{L}})}{|\Lambda_{\tilde{L}}|}
$$
for all $\tilde L$.
Then
$$
f(\tilde\rho,T) \leq \liminf_{\tilde{L} \to \infty} \frac{\tr(\cH \Gamma^D_{\tilde{L}})-TS(\Gamma^D_{\tilde{L}}) }{|\Lambda_{\tilde{L}}|}.
$$
\end{lemma}
\begin{proof}
We use the equivalence of the canonical and grand-canonical ensemble,
$$
f(\rho,T) = \sup_{\mu \in \R}\{\rho \mu - p(\mu,T)\}, 
$$
where the (grand-canonical) pressure is defined as 
$$
p(\mu, T) := -\lim_{\tilde{L} \to \infty} |\Lambda_{\tilde{L}}|^{-1} \inf\left\{\tr((\cH-\mu\cN)\Gamma)-TS(\Gamma) \mid \Gamma:\cF \to \cF, \Gamma\geq 0, \tr \Gamma =1\right\},
$$
see, e.g., \cite{ruelle1969statistical}.
We obtain
$$
-p(\mu,T) \leq \liminf_{\tilde L\to \infty} \frac{\tr(\cH \Gamma^D_{\tilde{L}})-TS(\Gamma^D_{\tilde{L}})}{|\Lambda_{\tilde L}|} - \mu \tilde{\rho}
$$
so that 
$$
f(\tilde{\rho},T) \leq \liminf_{\tilde L\to \infty} \frac{\tr(\cH \Gamma^D_{\tilde{L}})-TS(\Gamma^D_{\tilde{L}})}{|\Lambda_{\tilde L}|}.
$$
\end{proof}

\begin{proof}[Proof of \Cref{prop:Localization}]
This is a direct application of \Cref{lem:Dirichlet} and \Cref{from grand canoncial to canonical} with the choice $\tilde{L} = t(L+2\ell +R)$ and $t\to\infty.$
\end{proof}

\paragraph{Acknowledgments} This work was partially funded by the Deutsche Forschungsgemeinschaft (DFG, German Research Foundation) – Project-ID 470903074 – TRR 352.

\printbibliography

@book{ruelle1969statistical,
	location = {New York},
	title = {Statistical mechanics: Rigorous results},
	series = {Mathematical physics monograph series},
	publisher = {W. A. Benjamin},
	author = {Ruelle, David},
	date = {1969},
}

@article{Schlein_2008_UpperSoft,
	title = {The ground state energy of a low density Bose gas: a second order upper bound},
	volume = {78},
	issn = {1050-2947, 1094-1622},
	url = {http://arxiv.org/abs/0806.4873},
	doi = {10.1103/PhysRevA.78.053627},
	shorttitle = {The ground state energy of a low density Bose gas},
	pages = {053627},
	number = {5},
	journaltitle = {Physical Review A},
	shortjournal = {Phys. Rev. A},
	author = {Erdős, László and Schlein, Benjamin and Yau, Horng-Tzer},
	urldate = {2025-11-07},
	date = {2008-11-19},
	langid = {english},
	eprinttype = {arxiv},
	eprint = {0806.4873 [math-ph]},
}

@article{Mora_2003,
	title = {Extension of Bogoliubov theory to quasicondensates},
	volume = {67},
	rights = {http://link.aps.org/licenses/aps-default-license},
	issn = {1050-2947, 1094-1622},
	url = {https://link.aps.org/doi/10.1103/PhysRevA.67.053615},
	doi = {10.1103/PhysRevA.67.053615},
	pages = {053615},
	number = {5},
	journaltitle = {Physical Review A},
	shortjournal = {Phys. Rev. A},
	author = {Mora, Christophe and Castin, Yvan},
	urldate = {2026-04-02},
	date = {2003-05-30},
	langid = {english},
}

@article{JMKosterlitz_1973,
	title = {Ordering, metastability and phase transitions in two-dimensional systems},
	volume = {6},
	issn = {0022-3719},
	url = {https://iopscience.iop.org/article/10.1088/0022-3719/6/7/010},
	doi = {10.1088/0022-3719/6/7/010},
	pages = {1181--1203},
	number = {7},
	journaltitle = {Journal of Physics C: Solid State Physics},
	shortjournal = {J. Phys. C: Solid State Phys.},
	author = {Kosterlitz, J M and Thouless, D J},
	urldate = {2026-04-02},
	date = {1973-04-12},
}

@article{Andersen2002,
	title = {Ground state pressure and energy density of an interacting homogeneous Bose gas in two dimensions},
	volume = {28},
	rights = {http://www.springer.com/tdm},
	issn = {1434-6028, 1434-6036},
	url = {http://link.springer.com/10.1140/epjb/e2002-00242-6},
	doi = {10.1140/epjb/e2002-00242-6},
	pages = {389--396},
	number = {4},
	journaltitle = {The European Physical Journal B},
	shortjournal = {Eur. Phys. J. B},
	author = {Andersen, J.O.},
	urldate = {2026-04-02},
	date = {2002-08},
	langid = {english},
}

@article{Cherny2D,
	title = {Dilute Bose gas in two dimensions: Density expansions and the Gross-Pitaevskii equation},
	volume = {64},
	rights = {http://link.aps.org/licenses/aps-default-license},
	issn = {1063-651X, 1095-3787},
	url = {https://link.aps.org/doi/10.1103/PhysRevE.64.027105},
	doi = {10.1103/PhysRevE.64.027105},
	shorttitle = {Dilute Bose gas in two dimensions},
	pages = {027105},
	number = {2},
	journaltitle = {Physical Review E},
	shortjournal = {Phys. Rev. E},
	author = {Cherny, A. Yu. and Shanenko, A. A.},
	urldate = {2026-04-02},
	date = {2001-07-27},
	langid = {english},
}

@article{solovej2025mathematicalphysicsdilutebose,
	title = {Mathematical physics of dilute Bose gases},
	volume = {26},
	rights = {https://creativecommons.org/licenses/by/4.0/},
	issn = {1631-0705, 1878-1535},
	url = {https://comptes-rendus.academie-sciences.fr/physique/articles/10.5802/crphys.247/},
	doi = {10.5802/crphys.247},
	pages = {339--348},
	issue = {G1},
	journaltitle = {Comptes Rendus. Physique},
	author = {Solovej, Jan Philip},
	urldate = {2026-04-02},
	date = {2025-04-08},
	langid = {english},
}

@article{Mermin_Wagner_1966,
	title = {Absence of Ferromagnetism or Antiferromagnetism in One- or Two-Dimensional Isotropic Heisenberg Models},
	volume = {17},
	rights = {http://link.aps.org/licenses/aps-default-license},
	issn = {0031-9007},
	url = {https://link.aps.org/doi/10.1103/PhysRevLett.17.1133},
	doi = {10.1103/PhysRevLett.17.1133},
	pages = {1133--1136},
	number = {22},
	journaltitle = {Physical Review Letters},
	shortjournal = {Phys. Rev. Lett.},
	author = {Mermin, N. D. and Wagner, H.},
	urldate = {2026-03-25},
	date = {1966-11-28},
	langid = {english},
}

@misc{haberbergerUpperBound2024,
	title = {Upper bound for the free energy of dilute bose gases at low temperature},
	url = {https://arxiv.org/abs/2405.03378},
	author = {Haberberger, Florian and Hainzl, Christian and Schlein, Benjamin and Triay, Arnaud},
	date = {2024},
	eprinttype = {arxiv},
	eprint = {2405.03378 [math-ph]},
}

@book{vladimirov1971equations,
	location = {New York},
	edition = {Illustrated reprint},
	title = {Equations of mathematical physics},
	volume = {3},
	isbn = {978-0-8247-1713-1},
	series = {Monographs and textbooks in pure and applied mathematics},
	publisher = {M.Dekker},
	author = {Vladimirov, Vasil Sergeevich},
	date = {1971},
	note = {Pages: 418},
}

@book{nielsen_quantum_2012,
	edition = {1},
	title = {Quantum Computation and Quantum Information: 10th Anniversary Edition},
	rights = {https://www.cambridge.org/core/terms},
	isbn = {978-1-107-00217-3},
	url = {https://www.cambridge.org/core/product/identifier/9780511976667/type/book},
	doi = {10.1017/CBO9780511976667},
	shorttitle = {Quantum Computation and Quantum Information},
	publisher = {Cambridge University Press},
	author = {Nielsen, Michael A. and Chuang, Isaac L.},
	urldate = {2026-01-30},
	date = {2012-06-05},
}

@article{seiringer_thermodynamic_2006,
	title = {The Thermodynamic Pressure of a Dilute Fermi Gas},
	volume = {261},
	rights = {http://www.springer.com/tdm},
	issn = {0010-3616, 1432-0916},
	url = {http://link.springer.com/10.1007/s00220-005-1433-3},
	doi = {10.1007/s00220-005-1433-3},
	pages = {729--757},
	number = {3},
	journaltitle = {Communications in Mathematical Physics},
	shortjournal = {Commun. Math. Phys.},
	author = {Seiringer, Robert},
	urldate = {2026-01-12},
	date = {2006-02},
	langid = {english},
}

@article{Bose1924,
	title = {Plancks Gesetz und Lichtquantenhypothese},
	volume = {26},
	url = {https://doi.org/10.1007/BF01327326},
	pages = {178--181},
	journaltitle = {Zeitschrift Fur Physik},
	shortjournal = {Z. Phys.},
	author = {Bose, S. N.},
	date = {1924},
}

@book{robinson_thermodynamic_1971,
	location = {Berlin, Heidelberg},
	title = {The Thermodynamic Pressure in Quantum Statistical Mechanics},
	volume = {9},
	rights = {http://www.springer.com/tdm},
	isbn = {978-3-540-05640-9},
	url = {http://link.springer.com/10.1007/3-540-05640-8},
	doi = {10.1007/3-540-05640-8},
	series = {Lecture Notes in Physics},
	publisher = {Springer Berlin Heidelberg},
	author = {Robinson, Derek W.},
	urldate = {2025-11-12},
	date = {1971},
	langid = {english},
}

@article{Seiringer2008_FreeEnergy_3D_Lower,
	title = {Free energy of a dilute bose gas: Lower bound},
	volume = {279},
	issn = {1432-0916},
	url = {https://doi.org/10.1007/s00220-008-0428-2},
	doi = {10.1007/s00220-008-0428-2},
	pages = {595--636},
	number = {3},
	journaltitle = {Communications in Mathematical Physics},
	shortjournal = {Commun. Math. Phys.},
	author = {Seiringer, Robert},
	date = {2008},
}

@article{FGJMO-22,
	title = {The ground state energy of a two-dimensional bose gas},
	volume = {405},
	issn = {1432-0916},
	url = {https://doi.org/10.1007/s00220-023-04907-2},
	doi = {10.1007/s00220-023-04907-2},
	pages = {59},
	number = {3},
	journaltitle = {Communications in Mathematical Physics},
	shortjournal = {Commun. Math. Phys.},
	author = {Fournais, Søren and Girardot, Theotime and Junge, Lukas and Morin, Leo and Olivieri, Marco},
	date = {2024-02},
}

@article{Fannes_1973,
	title = {A continuity property of the entropy density for spin lattice systems},
	volume = {31},
	url = {https://doi.org/10.1007/BF01646490},
	shorttitle = {Commun. Math. Phys.},
	pages = {291--294},
	journaltitle = {Communications in Mathematical Physics},
	author = {Fannes, Mark},
	date = {1973},
}

@article{Basti_2023_Hard_GP,
	title = {A Second Order Upper Bound for the Ground State Energy of a Hard-Sphere Gas in the Gross–Pitaevskii Regime},
	volume = {399},
	issn = {1432-0916},
	url = {https://doi.org/10.1007/s00220-022-04547-y},
	doi = {10.1007/s00220-022-04547-y},
	pages = {1--55},
	number = {1},
	journaltitle = {Communications in Mathematical Physics},
	shortjournal = {Commun. Math. Phys.},
	author = {Basti, Giulia and Cenatiempo, Serena and Olgiati, Alessandro and Pasqualetti, Giulio and Schlein, Benjamin},
	date = {2023-04-01},
}

@article{BCGOPS_Hard_2024,
	title = {Upper bound for the ground state energy of a dilute bose gas of hard spheres},
	volume = {248},
	number = {6},
	journaltitle = {Archive for Rational Mechanics and Analysis},
	shortjournal = {Arch. Ration. Mech. Anal.},
	author = {Basti, Giulia and Cenatiempo, Serena and Giuliani, Alessandro and Olgiati, Alessandro and Pasqualetti, Giulio and Schlein, Benjamin},
	date = {2024},
}

@article{caraciExcitationSpectrumTwodimensional2023,
	title = {The excitation spectrum of two-dimensional bose gases in the gross–pitaevskii regime},
	volume = {24},
	url = {https://doi.org/10.1007/s00023-023-01278-1},
	doi = {10.1007/s00023-023-01278-1},
	pages = {2877--2928},
	number = {8},
	journaltitle = {Annales Henri Poincaré},
	shortjournal = {Ann. Henri Poincaré},
	author = {Caraci, Cristina and Cenatiempo, Serena and Schlein, Benjamin},
	date = {2023},
}

@misc{FGJMOT-24,
	title = {The free energy of dilute Bose gases at low temperatures interacting via strong potentials},
	url = {https://arxiv.org/abs/2408.14222},
	author = {Fournais, Søren and Girardot, Theotime and Junge, Lukas and Morin, Leo and Olivieri, Marco and Triay, Arnaud},
	date = {2024},
	eprinttype = {arxiv},
	eprint = {2408.14222 [math-ph]},
}

@misc{haberberger2024freeenergydilutebose,
	title = {The free energy of dilute Bose gases at low temperatures},
	url = {https://arxiv.org/abs/2304.02405},
	author = {Haberberger, Florian and Hainzl, Christian and Nam, Phan Thành and Seiringer, Robert and Triay, Arnaud},
	date = {2023},
	eprinttype = {arxiv},
	eprint = {2304.02405 [math-ph]},
}

@misc{Deuchert_2025_newupperboundfree,
	title = {A new upper bound on the specific free energy of dilute Bose gases},
	url = {http://arxiv.org/abs/2507.20877},
	doi = {10.48550/arXiv.2507.20877},
	number = {{arXiv}:2507.20877},
	publisher = {{arXiv}},
	author = {Basti, Giulia and Boccato, Chiara and Cenatiempo, Serena and Deuchert, Andreas},
	urldate = {2025-11-10},
	date = {2025-07-28},
	langid = {english},
	eprinttype = {arxiv},
	eprint = {2507.20877 [math-ph]},
}

@article{Einstein1925,
	title = {Quantentheorie des einatomigen idealen Gases. Zweite Abhandlung.},
	url = {https://doi.org/10.1007/BF01327507},
	pages = {3--14},
	journaltitle = {Sitzungsberichte der Preußischen Akademie der Wissenschaften, I},
	author = {Einstein, A.},
	date = {1925},
}

@article{Yin2010_Free_UpperBound,
	title = {Free energies of dilute bose gases: Upper bound},
	volume = {141},
	issn = {1572-9613},
	url = {https://doi.org/10.1007/s10955-010-0066-x},
	doi = {10.1007/s10955-010-0066-x},
	pages = {683--726},
	number = {4},
	journaltitle = {Journal of Statistical Physics},
	author = {Yin, Jun},
	date = {2010},
}

@article{Dyson_Hard_1957,
	title = {Ground-state energy of a hard-sphere gas},
	volume = {106},
	url = {https://link.aps.org/doi/10.1103/PhysRev.106.20},
	doi = {10.1103/PhysRev.106.20},
	pages = {20--26},
	number = {1},
	journaltitle = {Physical Review},
	shortjournal = {Phys. Rev.},
	publisher = {American Physical Society},
	author = {Dyson, F. J.},
	date = {1957-04},
	note = {Number of pages: 0},
}

@article{basti2021newsecondorderupper,
	title = {A new second-order upper bound for the ground state energy of dilute Bose gases},
	volume = {9},
	issn = {2050-5094},
	url = {https://www.cambridge.org/core/product/identifier/S2050509421000669/type/journal_article},
	doi = {10.1017/fms.2021.66},
	pages = {e74},
	journaltitle = {Forum of Mathematics, Sigma},
	shortjournal = {Forum of Mathematics, Sigma},
	author = {Basti, Giulia and Cenatiempo, Serena and Schlein, Benjamin},
	urldate = {2025-10-29},
	date = {2021},
	langid = {english},
}

@book{GreenBook,
	title = {The mathematics of the Bose gas and its condensation},
	series = {Oberwolfach seminars},
	publisher = {Birkhäuser},
	author = {Lieb, E. H. and Seiringer, R. and Solovej, J. P. and Yngvason, J.},
	date = {2005},
}

@article{Lieb2001,
	title = {The ground state energy of a dilute two-dimensional bose gas},
	volume = {103},
	issn = {1572-9613},
	url = {https://doi.org/10.1023/A:1010337215241},
	doi = {10.1023/A:1010337215241},
	pages = {509--526},
	number = {3},
	journaltitle = {Journal of Statistical Physics},
	author = {Lieb, Elliott H. and Yngvason, Jakob},
	date = {2001-05},
}

@article{Mora2D,
	title = {Ground state energy of the two-dimensional weakly interacting bose gas: First correction beyond bogoliubov theory},
	volume = {102},
	url = {https://link.aps.org/doi/10.1103/PhysRevLett.102.180404},
	doi = {10.1103/PhysRevLett.102.180404},
	pages = {180404},
	number = {18},
	journaltitle = {Physical Review Letters},
	shortjournal = {Phys. Rev. Lett.},
	publisher = {American Physical Society},
	author = {Mora, Christophe and Castin, Yvan},
	date = {2009-05},
	note = {Number of pages: 4},
}

@article{DEUCHERT_MAYER_SEIRINGER_2020,
	title = {The free energy of the two-dimensional dilute Bose gas. I. Lower bound},
	volume = {8},
	url = {https://doi.org/10.1017/fms.2020.17},
	doi = {10.1017/fms.2020.17},
	pages = {e20},
	journaltitle = {Forum of Mathematics, Sigma},
	author = {Deuchert, Andreas and Mayer, Simon and Seiringer, Robert},
	date = {2020},
}

@article{Jastrow1955,
	title = {Many-body problem with strong forces},
	volume = {98},
	url = {https://link.aps.org/doi/10.1103/PhysRev.98.1479},
	doi = {10.1103/PhysRev.98.1479},
	pages = {1479--1484},
	number = {5},
	journaltitle = {Physical Review},
	shortjournal = {Phys. Rev.},
	publisher = {American Physical Society},
	author = {Jastrow, Robert},
	date = {1955-06},
	note = {Number of pages: 0},
}

@article{Mayer2020,
	title = {The free energy of the two-dimensional dilute Bose gas. {II}. Upper bound},
	volume = {61},
	issn = {0022-2488},
	url = {https://doi.org/10.1063/5.0005950},
	doi = {10.1063/5.0005950},
	pages = {061901},
	number = {6},
	journaltitle = {Journal of Mathematical Physics},
	author = {Mayer, Simon and Seiringer, Robert},
	date = {2020-06},
	note = {tex.eprint: https://pubs.aip.org/aip/jmp/article-pdf/doi/10.1063/5.0005950/13436963/061901{\textbackslash}\_1{\textbackslash}\_online.pdf},
}

@article{Hohenberg1967,
	title = {Existence of long-range order in one and two dimensions},
	volume = {158},
	url = {https://link.aps.org/doi/10.1103/PhysRev.158.383},
	doi = {10.1103/PhysRev.158.383},
	pages = {383--386},
	number = {2},
	journaltitle = {Physical Review},
	shortjournal = {Phys. Rev.},
	publisher = {American Physical Society},
	author = {Hohenberg, P. C.},
	date = {1967-06},
	note = {Number of pages: 0},
}

@article{Berezinskii:1972fet,
    author = "Berezinskii, V. L.",
title = {Destruction of Long-range Order in One-dimen\-sional and Two-dimen\-sional Systems Possessing a Continuous Symmetry Group. II. Quantum Systems},
journal = "Sov. Phys. JETP",
    volume = "34",
    number = "3",
    pages = "610--616",
    year = "1972"
}

@article{Schick,
	title = {Two-dimensional system of hard-core bosons},
	volume = {3},
	url = {https://link.aps.org/doi/10.1103/PhysRevA.3.1067},
	doi = {10.1103/PhysRevA.3.1067},
	pages = {{1067--1073}},
	number = {3},
	journaltitle = {Physical Review A: Atomic, Molecular, and Optical Physics},
	shortjournal = {Phys. Rev. A},
	publisher = {American Physical Society},
	author = {Schick, M.},
	date = {1971-03},
	note = {Number of pages: 0},
}

\end{document}